\documentclass[lineno]{custom-arxiv}

\nolinenumbers

\usepackage{amsmath}
\usepackage{newtxtext}
\usepackage[subscriptcorrection]{newtxmath}
\usepackage[plain,noend]{algorithm2e}

\makeatletter

\usepackage{algorithmic}

\usepackage{ifthen}
\usepackage{pstricks}
\usepackage{subfigure}
\usepackage{enumerate}
\usepackage{ulem} 
\usepackage{verbatim}
\usepackage{multirow}
\usepackage{tabularx}
\usepackage{url}

\usepackage{xcolor,colortbl}

\let\emptyset\varnothing

\begin{document}



\markboth{Valdez-Cabrera \& Willis}{Distances between Extension Spaces}

\title{\vspace{-1cm}Distances between Extension Spaces of Phylogenetic Trees}


\author{MAR\'IA A. VALDEZ-CABRERA}
\affil{Department of Biostatistics, University of Washington,\\ Seattle, WA, U.S.A.
\email{mariavc@uw.edu}}

\author{\and AMY D. WILLIS}
\affil{Department of Biostatistics, University of Washington,\\ Seattle, WA, U.S.A.
\email{adwillis@uw.edu}}

\maketitle






\newboolean{DEBUG}
\setboolean{DEBUG}{false}

\newrgbcolor{amycolor}{.5 .1 .99}
\ifthenelse {\boolean{DEBUG}}
{\newcommand{\amy}[1]{{\amycolor{[\@Amy: #1]}}}}
{\newcommand{\amy}[1]{}}

\newrgbcolor{mariacolor}{.392 .585 .93}
\ifthenelse {\boolean{DEBUG}}
{\newcommand{\maria}[1]{{\mariacolor{[\@María: #1]}}}}
{\newcommand{\maria}[1]{}}

\ifthenelse {\boolean{DEBUG}}
{\newcommand{\mariadraft}[1]{{\mariacolor{#1}}}}
{\newcommand{\mariadraft}[1]{#1}}

\begin{abstract}
    Phylogenetic trees summarize evolutionary relationships between organisms, and tools to analyze collections of phylogenetic trees enable contrasts between different genes' ancestry. The BHV metric space has enabled the analysis of collections of trees that share a common set of leaves, but many genes are not shared, even between closely related species. BHV extension spaces represent trees with non-identical leaf sets in a common BHV space, but limited analytical tools exist for extension spaces. We define the distance between two phylogenetic trees with non-identical leaf sets as the shortest BHV distance between their extension spaces, and develop a reduced gradient algorithm to compute this distance. We study the scalability of our algorithm and apply it to analyze gene trees spanning multiple domains of life. Our distance and algorithm offer a fully general, interpretable approach to analyzing both ancient and recent evolutionary divergence. 
\end{abstract}



\section{Introduction}

Phylogenetic trees, which describe shared ancestry between organisms through their branching structure (topology) and branch lengths, are a fundamental tool in biology for analyzing evolutionary relationships. In addition to genetically organizing life, modern phylogenetics enables the prediction of characteristics of uncultivated organisms \citep{RinkeChristian2013Iitp, ParksDonovanH2020Acdt}, lineage tracing of pathogens \citep{chen2016listeria}, forensic analysis \citep{scaduto2010source}, 
and a view into the origins of complex multicellular life \citep{imachi2020isolation, zhu2019phylogenomics}. Phylogenetics studies typically leverage genomic data that encompass multiple genes shared between organisms. However, due to natural biological processes like deep coalescence or horizontal gene transfer, phylogenetic trees that describe different genes frequently exhibit differing topologies and/or 
significantly different branch lengths \citep{Maddison:1997de}. However, as phylogenetic trees are complex data objects, their analysis requires sophisticated mathematical methodology. 

Formally, phylogenetic trees are connected, acyclic, edge-weighted graphs, with each leaf node representing an organism. A phylogenetic tree with leaf labels $\mathcal{N}$ is an element of a BHV metric space \citep{BILLERA2001}. The BHV space $\left(\mathcal{T}^{\mathcal{N}}, d\right)$ is a Hadamard space \citep[Definition 1.1]{Bridson:1999ky}, which ensures the existence of a unique geodesic between any two trees, and the metric $d$ is the length of this geodesic. This complete metric space with continuous paths has enabled powerful analytical tools for the analysis of collections of phylogenetic trees, including means and variances  \citep{blow, bbb, willis2018uncertainty, brown2020mean}, confidence sets \citep{willis2019confidence}, density estimation \citep{weyenberg2014kdetrees}, and clustering \citep{gori2016clustering}. 

One limitation of BHV-based tools is that they are restricted to analyzing trees on the same set of organisms. This substantially limits the usefulness of these tools in practice, especially when analyzing collections of organisms that experience high rates of gene loss and transfer. To address this limitation, \cite{GrindstaffGillian2019RoPL} introduced  \textit{extension spaces}. Informally, for a tree $T$ with leaves $\mathcal{L} \subseteq \mathcal{N}$, the elements of the extension space $E^{\mathcal{N}}_{T} \subset \mathcal{T}^{\mathcal{N}}$ are the trees created by attaching the leaves in $\mathcal{N} \setminus \mathcal{L}$ to $T$ in every possible way. Extension spaces provide a path towards analyzing trees with non-identical leaf sets inside the same BHV space, because $\mathcal{N}$ can be constructed as the union of all trees' leaf sets. 

In addition to introducing extension spaces, \cite{GrindstaffGillian2019RoPL} proposed two measures of ``compatibility'' between trees with non-identical leaf sets. These measures are based on constructing neighborhoods around the trees' extension spaces and determining the smallest radius for these neighborhoods to intersect. Trees in these intersections can be seen as ``supertrees'' that combine the two phylogenies. However, the neighborhoods will never intersect if the extension spaces do not share at least one common topology. Therefore, these measures are only defined for trees that are topologically compatible. 


To address the need for a compatibility measure that is defined for \textit{any} pair of trees, we consider the shortest BHV path between their extension spaces \citep[Section 3.4]{GrindstaffGillian2019RoPL}. That is, we define the shortest distance between two extension spaces as follows:
\begin{equation}
        d\left(E_{T_1}^{\mathcal{N}
}, E_{T_2}^{\mathcal{N}}\right) := \inf_{(t_1,t_2) \in E_{T_1}^{\mathcal{N}
} \times E_{T_2}^{\mathcal{N}
}} d(t_1, t_2).
\label{eq:distanceDefinition}
\end{equation}
In this paper, we propose the first algorithm to compute this distance, leveraging the polynomial-time BHV distance algorithm developed by \cite{OwenMegan2011}. 

This paper is structured as follows. We begin with overviews of BHV tree space in Section \ref{sec:Phylogenetic Tree} and extension spaces in Section \ref{sec:ExtensionSpace}. We then review necessary algorithmic background in Section \ref{sec:OptimizationOverview}. Our main results and our proposed algorithm are detailed in Section \ref{sec:ESOdistance}, culminating in Algorithm \ref{alg:OrthantExtensionDistance}. We illustrate the runtime and scalability of our method using simulated data examples in Section \ref{sec:Experiments} and demonstrate the applicability of our algorithm with a comparison of two gene trees in Section \ref{sec:RealDataAnalysis}. We conclude with a discussion in Section \ref{sec:conclusions}. Additional details pertaining to our algorithm can be found in Section \ref{sec:AlgDets}.


\section{Background}
\label{sec:Phylogenetic Tree}

A phylogenetic tree is a graphical representation of the evolutionary history of a set of organisms. Its topological shape describes divergence events, and the lengths of its edges indicate the distance (e.g., time or number of genetic mutations) between these events.  

\begin{definition}
    \label{def:PhyloTree}
    A \textbf{phylogenetic tree} $T$ on a set of organisms $\mathcal{L}$ is a connected weighted acyclic graph whose \textbf{leaves} (nodes of degree 1) are labelled by the elements of $\mathcal{L}$, and all interior nodes are at least of degree 3. All edge-weights are non-negative values referred to as \textbf{lengths}. The branching pattern of the graph gives the \textbf{topology} of the tree. 
\end{definition}


\begin{remark}
    Sometimes, phylogenetic trees have a node that represents the \textbf{root} of the tree, that is, a common ancestor to all organisms in $\mathcal{L}$. Throughout this paper, we will discuss unrooted trees (Definition \ref{def:PhyloTree}), but our results apply to rooted trees without loss of generality. 
\end{remark}

Removing an edge from a tree $T$ divides it into two connected graphs, each with a subset of its leaves. Therefore, each edge in $T$ can be uniquely identified by the partition $\mathcal{L} = \mathcal{G} \sqcup \{\mathcal{L} \setminus \mathcal{G}\}$ induced when removed. 
Throughout we assume $\mathcal{G}$ to be the smaller subset ($|\mathcal{G}| \leq |\mathcal{L} \setminus \mathcal{G}|$) and denote the partition by $\mathcal{G} \mid \mathcal{L} \setminus \mathcal{G}$. We call the partition identifying an edge its \textbf{split}, and we use edge and split interchangeably.  Note that $|\mathcal{G}| = 1$ for external edges, and $2 \leq |\mathcal{G}|$ for internal edges. Edges that induce the same partition on topologically distinct trees are considered to be the same edge. 

The topology of a tree is fully defined by its internal splits. We use $\mathcal{S}(T)$ to refer to the set of internal splits in $T$ (equivalently, the topology of $T$). For a tree with $l$ leaves, the cardinality of this set falls in $\{0, 1, \ldots, l-3\}$, 
and the tree is fully resolved (binary) when it has $l-3$ internal edges. Additionally, $T$ has $l$ external edges, one per leaf. We use $\mathcal{H}(T)$ to refer to the set of external edges and $\mathcal{P}(T)$ to the set of all edges of $T$; so $ \mathcal{P}(T) = \mathcal{S}(T) \cup \mathcal{H}(T)$.

\subsection{BHV tree space}
\label{sec:TreeSpace}
\cite{BILLERA2001} proposed a complete, geodesic metric space of non-positive curvature for phylogenetic trees that share a leaf set. The metric distance, commonly referred to as the \textit{BHV distance}, captures both topological and edge-length differences between two trees. Unique geodesics connect trees in BHV space  \citep[Definition 1.3]{Bridson:1999ky}, facilitating the development of BHV-based methods for analyzing collections of phylogenetic trees \citep{bbb,nye11,weyenberg2014kdetrees,gori2016clustering,willis2018uncertainty,willis2019confidence, nyepca,brown2020mean,teichman2023analyzing}. This paper presents an algorithm to find distances between subspaces of the BHV space, where these subspaces are generated by trees with some absent leaves. 

Consider all possible topologies for trees with leaf set $\mathcal{N}$ of size $n = |\mathcal{N}|$. For each topology, order its internal splits $S = \{s_1, s_2, \hdots, s_m\}$, where $s_i = \mathcal{G}_i \mid \mathcal{N} \setminus \mathcal{G}_i$. Ordering can be assigned to the splits in any consistent way; for example, an order can be assigned to the leaves in $\mathcal{N}$ and the edges may inherit the lexicographic order on the subsets $\left\{\mathcal{G}_i\right\}_{i=1}^{m}$. A tree with a given topology may be represented by a $(n+m)-$dimensional non-negative vector, where the first $n$ coordinates represent the lengths of the external edges and the last $m$ coordinates represent the lengths of the internal branches  (Fig.~1). Thus, all trees with topology $S$ can be represented by a $(n+m)-$dimensional non-negative vector, which together form an orthant in $\mathbb{R}^{n+m}$. We call this set of vectors the \textbf{topology orthant}, denoted by $\mathcal{O}(S)$  \citep[Section 2.1]{OwenMegan2011}. We also employ the notation $\mathcal{O}(T)$ to refer to the lowest-dimensional orthant containing $T$; that is, $\mathcal{O}(T) = \mathcal{O}(S(T))$. For a topology orthant $O$, we let $\mathcal{S}(O)$ denote the set of internal splits that define the topology of $O$.

\begin{figure}
    \label{fig:OneTopology}
    \centering
    \subfigure[]{\includegraphics[width=0.40\textwidth]{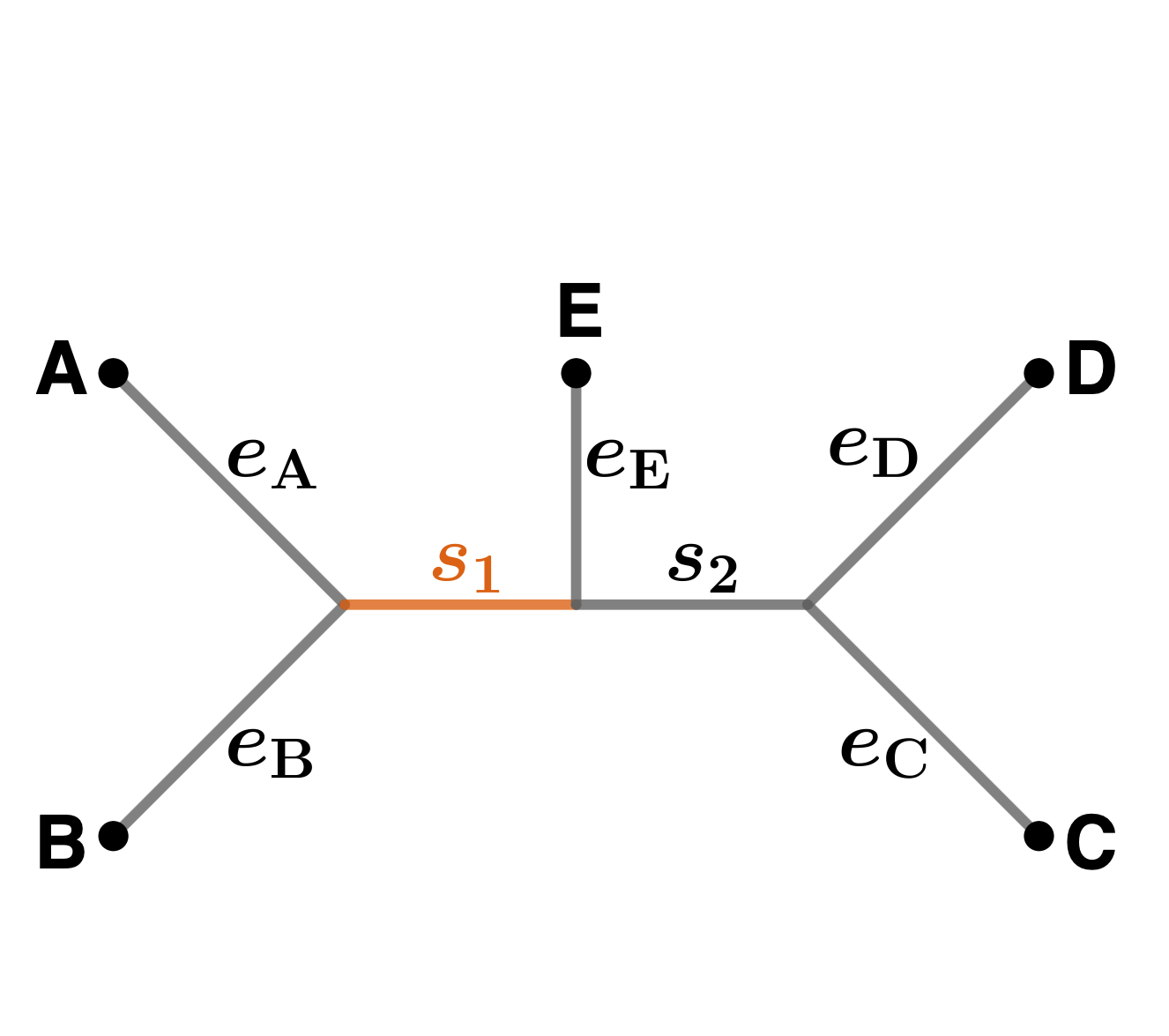}} 
    \subfigure[]{\includegraphics[width=0.46\textwidth]{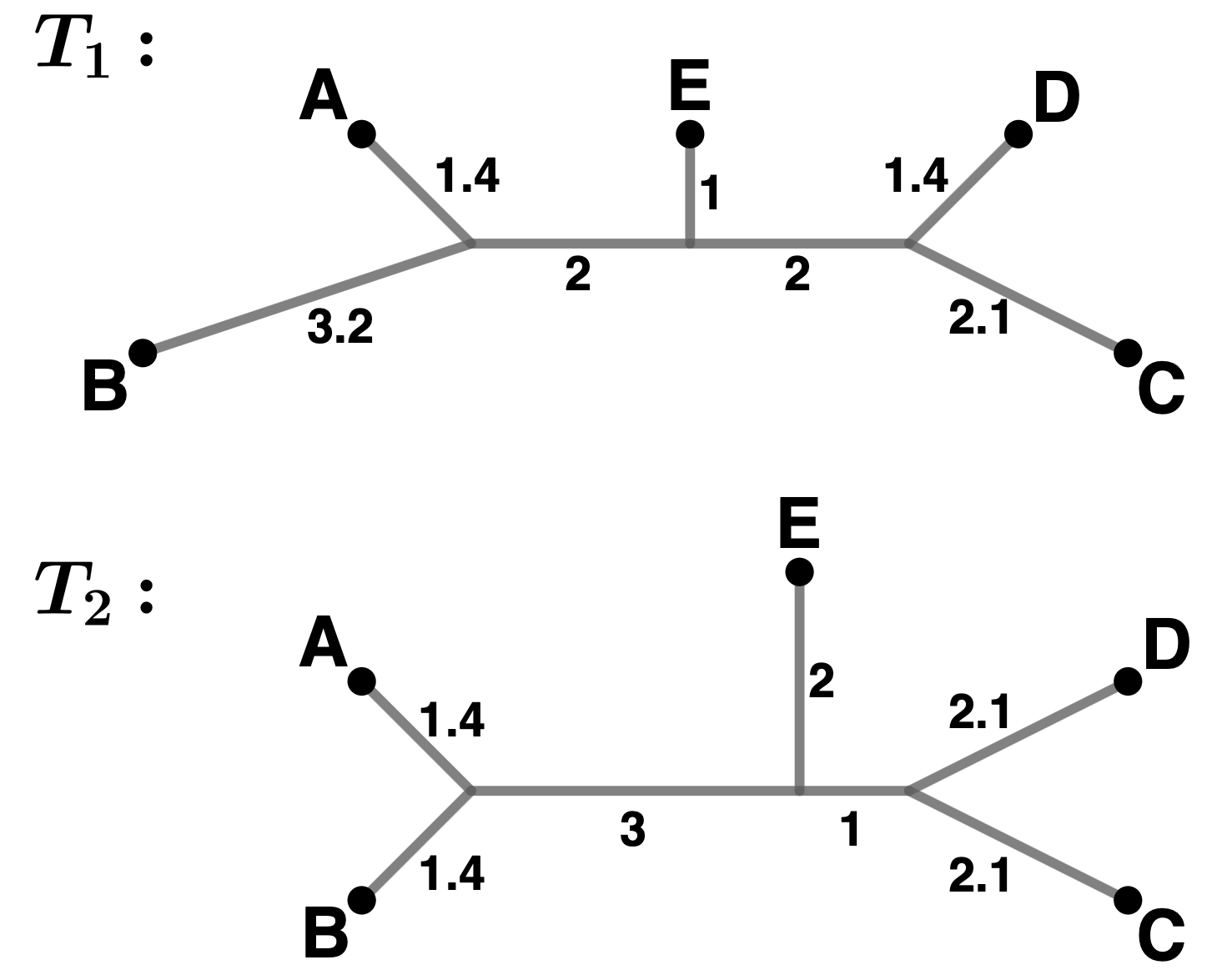}} 
    \caption{(a) A phylogenetic tree with leaf set $\mathcal{N} = \{A, B, C, D, E\}$ and internal split $s_1 = \{A,B\} \mid \{C,D,E\}$ highlighted in orange. This topology can be summarized by 7-dimensional vectors $(e_A, e_B, e_C, e_D, e_E, s_1, s_2)$. (b) Two different trees with the same topology. In its topology orthant, $T_1$ corresponds to the vector $(1.4, 3.2, 2.1, 1.4, 1, 2, 2)$ while $T_2$ corresponds to the vector $(1.4, 1.4, 2.1, 2.1, 2, 3, 1)$.}
\end{figure}

To connect the topology orthants, we ``glue'' them along appropriate equivalence classes. We consider any tree with an internal edge of length zero equivalent to the tree without that edge but with all other edges in common and of the same length (Fig.~2). Given a set of internal splits $S$ with $m$ elements, any topology orthant $\mathcal{O}(S')$ corresponding to a proper subset $S' \subset S$ of these internal edges can be viewed as a \textbf{face} of the larger topology orthant $\mathcal{O}(S)$  \citep[Section 2]{BILLERA2001}. 
We use the notation $O' \subset O$ to refer to $O'$ as a face of $O$. Two topology orthants $\mathcal{O}(S_1)$ and $\mathcal{O}(S_2)$ such that $S := S_1 \cap S_2 \neq \emptyset$ are glued along the face corresponding to the topology orthant $\mathcal{O}(S)$ (Fig.~2(b)). The only topology orthants that are not faces of another topology orthant are those of maximum-dimension $2n-3$, which correspond to topologies of binary trees. 

A path can be traced between any two trees, where a path is a piecewise continuous curve with each piece being fully contained within a single orthant. Moreover, BHV space is a complete geodesic space with non-positive curvature  \citep[Lemma 4.1]{BILLERA2001}, and thus the shortest path between any two trees (the \textbf{geodesic}) is unique. The \textbf{distance} between two trees $d(T_1, T_2)$ is the length of the geodesic connecting them. 
As a consequence of $(\mathcal{T}^{\mathcal{N}}, d)$ being a complete metric space of non-positive curvature, the function $d: \mathcal{T}^{\mathcal{N}} \times \mathcal{T}^{\mathcal{N}} \mapsto \mathbb{R}$ is convex on $\mathcal{T}^{\mathcal{N}} \times \mathcal{T}^{\mathcal{N}}$  \citep[Corollary 2.5]{sturm}. 

\begin{remark}
\label{Remark:euclideanOrthants}
    In the BHV space the geodesic between two trees in the same orthant is the line segment between them contained in the orthant. Thus, the BHV metric coincides with the Euclidean metric when restricted to a single orthant. An orthant can be viewed as a Euclidean subspace of $\mathcal{T}^{\mathcal{N}}$. This extends to the product of two orthants, i.e., $O_1 \times O_2$ is an Euclidean subspace of $\mathcal{T}^{\mathcal{N}} \times \mathcal{T}^{\mathcal{N}}$.
\end{remark}

\subsection{Geodesics in BHV space}
\label{sec:OwenProvan}

In addition to proposing a polynomial-time algorithm to compute geodesics in the BHV space, \cite{OwenMegan2011} give a closed-form expression for the length of these geodesics. This will be useful to describe how BHV distances are influenced by small changes to the endpoint trees.

\begin{definition}
    \label{Def:PathSpace}
     \citep[Definition 3.3]{owen1} For trees $T_1$ and $T_2$, consider the sets of internal splits $S_1 = \mathcal{S}(T_1)$, $S_2 = \mathcal{S}(T_2)$ and the set of common internal splits $C = S_1 \cap S_2$. For a sequence of subsets $S_1 = G_0 \supset G_1 \supset \hdots \supset G_{k-1} \supset G_{k} = C$ and $C = F_0 \subset F_1 \subset \hdots \subset F_{k-1} \subset F_{k} = S_2$ such that $O_i = \mathcal{O}(G_i \cup F_i)$ are valid topology orthants for all $i \in \{0,\hdots, k\}$, then the sequence of orthants,
\begin{equation*}
   \mathcal{O}(T_1) = O_0  \rightarrow O_1 \rightarrow \hdots \rightarrow O_k = \mathcal{O}(T_2),
\end{equation*}
is a \textbf{path space} from $T_1$ to $T_2$. 
\end{definition}

\begin{definition}
    \label{Def:Support}
     \citep[Section 2.3]{OwenMegan2011} Consider a path space $O_0 \rightarrow \hdots \rightarrow O_k$. When transitioning from $O_{i-1}$ to $O_{i}$, denote by $A_i = G_{i-1} \setminus G_{i}$ the set of dropped edges and by $B_i = F_{i} \setminus F_{i-1}$ the set of added edges. For ordered sets $\mathcal{A} = \left\{A_1, \hdots, A_k\right\}$ and $\mathcal{B} = \left\{B_1, \hdots, B_{k}\right\}$, the pair $(\mathcal{A}, \mathcal{B})$ gives the \textbf{support} of the path space. 
\end{definition}

\begin{figure}
    \label{fig:OrthantFaces}
    \centering
    \subfigure[]{\includegraphics[width=0.45\textwidth]{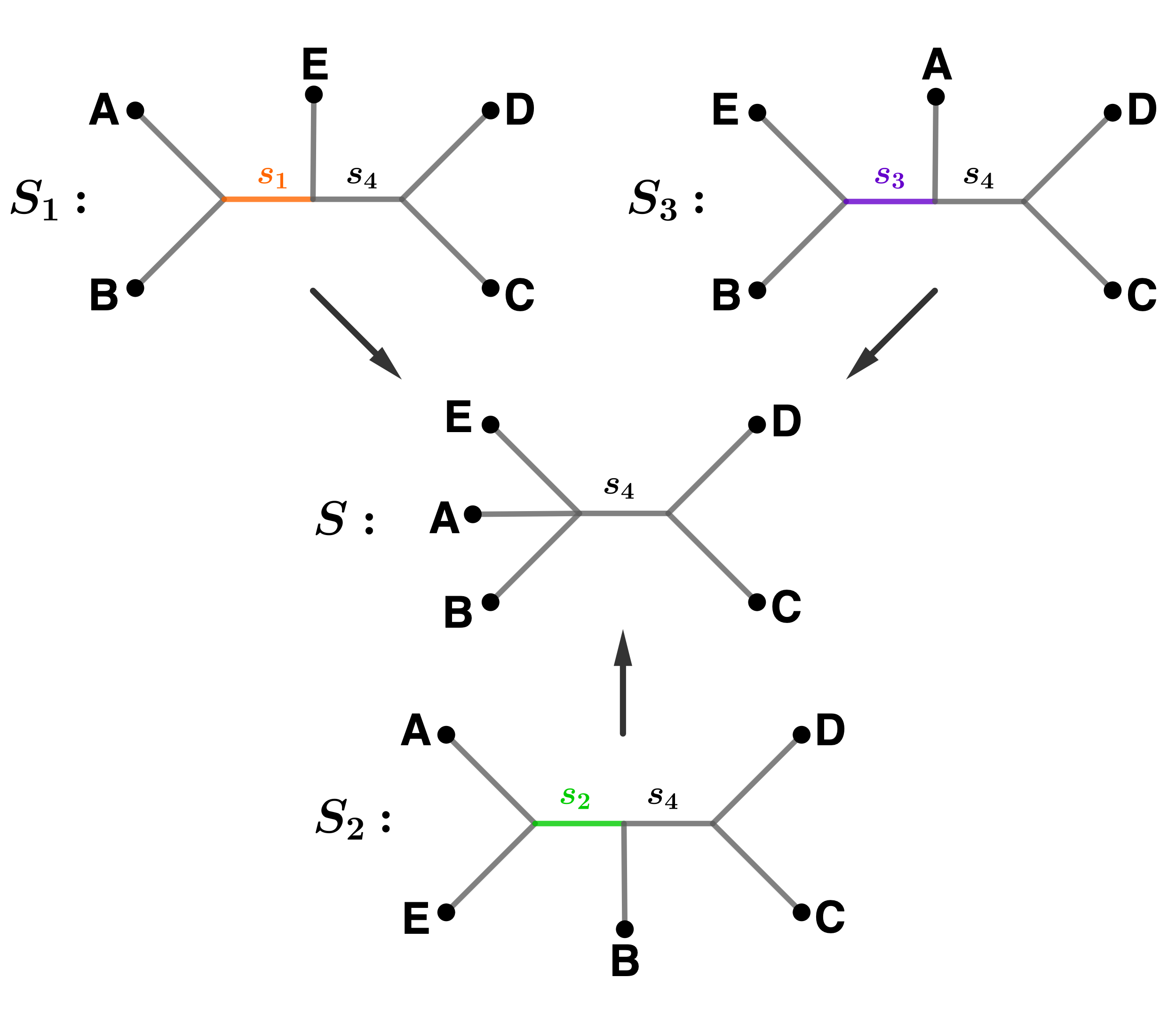}}
    \subfigure[]{\includegraphics[width=0.49\textwidth]{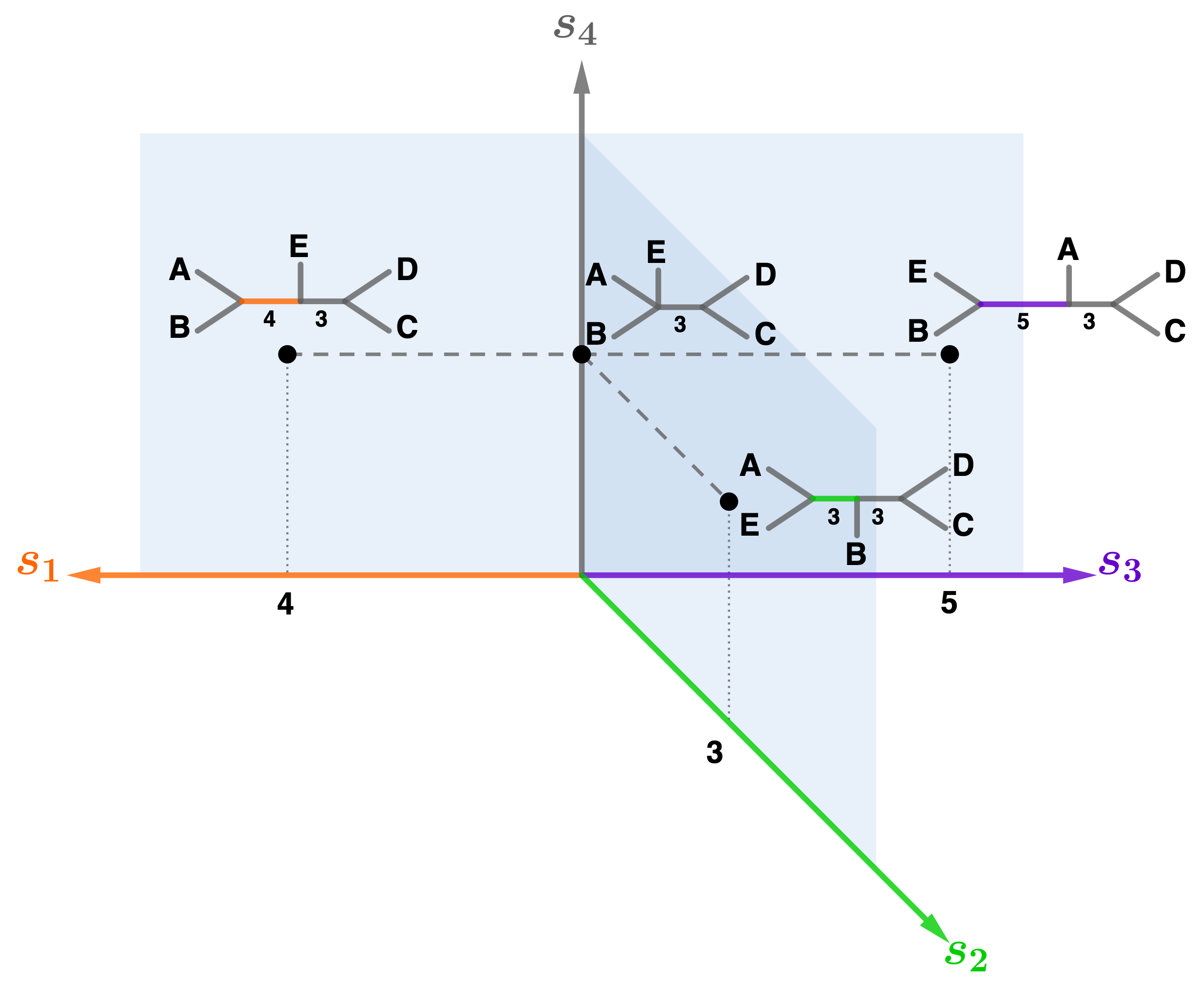}} 
    \caption{(a) Topologies $S_1$, $S_2$ and $S_3$ share topology $S$ as a face. 
    The topology $S_i$ is composed of internal edges $S_i = \{s_i, s_4\}$, while the topology $S$ only has $S = \{s_4\}$; thus $\mathcal{O}(S)$ is a face of $\mathcal{O}(S_i)$ for $i=1, 2, 3$. (b) The topology orthants $\mathcal{O}(S_1)$, $\mathcal{O}(S_2)$ and $\mathcal{O}(S_3)$ are ``glued'' at $\mathcal{O}(S)$. The dimensions arising from external edges are not shown.}
\end{figure}

A path space describes a way to move from the topology orthant of $T_1$ to the topology orthant of $T_2$ through connected orthants. A key property of path spaces is that the common edges in $C = S(T_1) \cap S(T_2)$ are in the topologies for all orthants in the path, while uncommon edges in $S_1$ (corresponding to the $A_i$'s) are gradually replaced by uncommon edges in $S_2$ (corresponding to the $B_i$'s). 

Definition \ref{Def:PathSpace} differs slightly from that of \cite[Definition 3.3]{owen1} in that it allows $T_1$ and $T_2$ to have common edges. This change is congruent with \cite[Proposition 4.1]{BILLERA2001}, which ensures the geodesic from $T_1$ to $T_2$ will be completely contained in a path space. The properties of the support of the path space that contains the geodesic were first explored in \cite[Section 2.3]{OwenMegan2011} and the \textit{Geodesic Treepath Problem} (GTP) algorithm  \citep[Section 4]{OwenMegan2011} finds this proper support in polynomial-time. 

Let $|p|_{T}$ be the length of the edge $p$ in tree $T$ and $||P||_{T} = \sqrt{\sum_{p \in P} |p|_T^{2}}$ be the $L^2$-norm of the lengths of all edges in the set of splits $P$ in $T$. The length of the geodesic will depend on the lengths of the uncommon edges through the $L^2$-norm of each $A_i$ and $B_i$ in the support; and on the difference in length of the common edges, including the external edges $H = \mathcal{H}(T_1) = \mathcal{H}(T_2)$. The length of the geodesic can be written in these terms by
\begin{equation}
\label{eq:FirstDistance}
        d(T_1, T_2) = \left|\left| \left(||A_1||_{T_1} + ||B_1||_{T_2}, \hdots, ||A_k||_{T_1} + ||B_k||_{T_2},  \left(|s|_{T_1} - |s|_{T_2}\right)_{s \in K}\right) \right|\right|,
\end{equation}
where 
$(a_i)_{i\in I}$ denotes the vector containing objects in $I$ as entries, and $K = C \cup H$  \citep[Theorem 1.2]{MillerOwenProvan} (see also \cite[Section 4]{OwenMegan2011}). 

\subsection{Extension spaces}
\label{sec:ExtensionSpace}

Extension spaces are a way to represent trees with a smaller leaf set $\mathcal{L} \subset \mathcal{N}$ in the BHV space $\mathcal{T}^{\mathcal{N}}$, thereby providing a framework for analyzing trees with non-identical leaf sets as subsets of $\mathcal{T}^{\mathcal{N}}$. 

A tree $T$ with leaf set $\mathcal{L}$ gives rise to a distance between any two members of $\mathcal{L}$, defined as the shortest path between them if traversing the graph $T$. We therefore consider the discrete metric space on $\mathcal{L}$ with metric $d_{T}(\ell_1, \ell_2)$ defined to be the distance between $\ell_1, \ell_2 \in \mathcal{L}$ in $T$ (Fig.~3(a)). 

\begin{figure}
    \label{fig:MetricTrees}
    \centering
    \subfigure[]{\includegraphics[height=4cm]{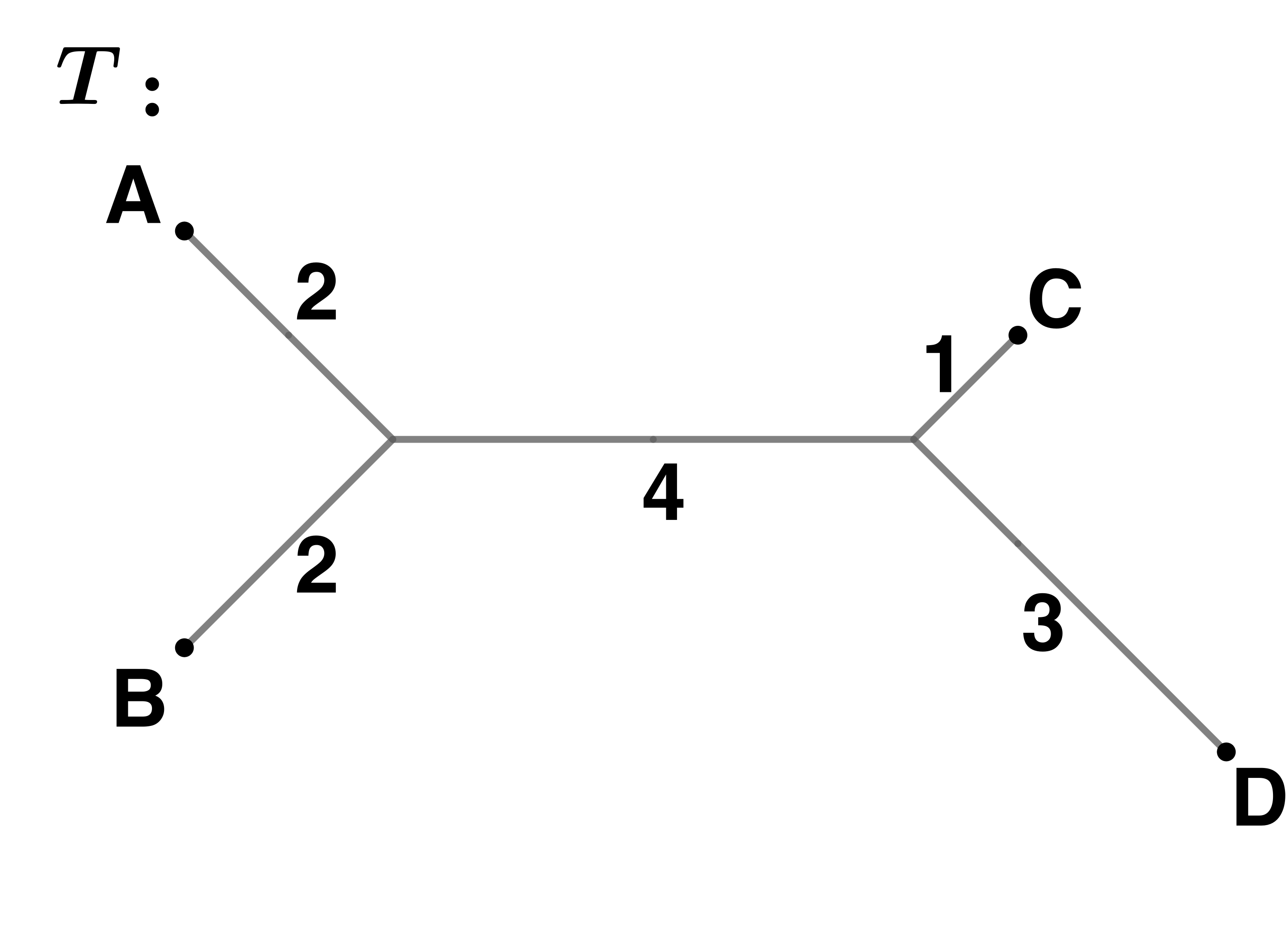}}
    \subfigure[]{\includegraphics[height=4cm]{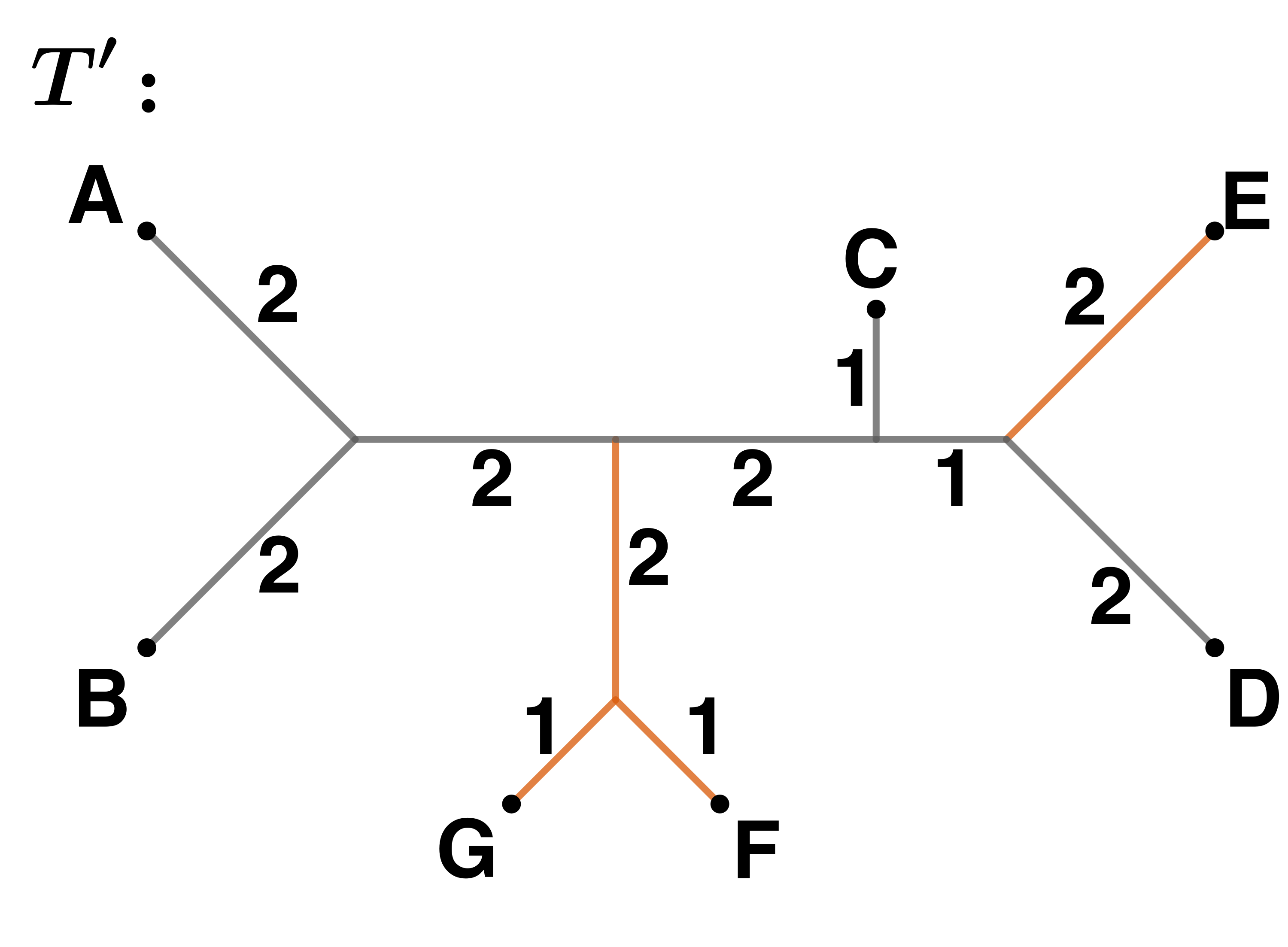}}
    \caption{For the complete leaf set $\mathcal{N} = \{A, B, C, D, E, F, G\}$ and subset $\mathcal{L} = \{A,B,C,D\}$. (a) A tree $T \in \mathcal{T}^{\mathcal{L}}$. For the discrete metric space over $\mathcal{L}$ given by $T$, $d_T(A,D) = 9$. (b) A tree $T'\in \mathcal{T}^{\mathcal{N}}$ in the extension space of $T$. The inconsequential edges for $E_{T}^{\mathcal{N}}$ are highlighted in orange.}
\end{figure}

\begin{definition}
    Given a tree $T$ with leaf set $\mathcal{L} \subseteq \mathcal{N}$, the \textbf{extension space} of $T$ in the BHV space $\mathcal{T}^\mathcal{N}$ is 
    \begin{equation}
        \label{def:extension}
        E_{T}^{\mathcal{N}} = \left\{T'\in \mathcal{T}^{\mathcal{N}} | d_{T'}(\ell_1, \ell_2) = d_{T}(\ell_1, \ell_2) \text{ for all } \ell_1, \ell_2 \in \mathcal{L}\right\}.
    \end{equation}
\end{definition}

\begin{definition}
    For $\mathcal{L} \subseteq \mathcal{N}$, the \textbf{tree dimensionality reduction} (TDR) map is the projection function $\Psi_{\mathcal{L}}: \mathcal{T}^{\mathcal{N}} \mapsto \mathcal{T}^{\mathcal{L}}$ that maps each tree $T'\in \mathcal{T}^{\mathcal{N}}$ to the \textbf{unique} tree $T \in \mathcal{T}^{\mathcal{L}}$ such that $(\mathcal{T}^{\mathcal{L}}, d_{T})$ is the metric subspace of $(\mathcal{T}^{\mathcal{N}}, d_{T'})$ restricted to $\mathcal{L}$. 
\end{definition}

Definition \ref{def:extension} was first introduced by  \citep[Definition 3.9]{GrindstaffGillian2019RoPL}, who proved it coincides with the pre-image of $T$ under the tree dimensionality reduction map $\Psi_{\mathcal{L}}$  \citep[Definition 4.1]{zairis2016genomic}. To find $\Psi_{\mathcal{L}}(T')$, we remove any edges that do not belong to a path connecting leaves in $\mathcal{L}$. We then traverse each induced degree-2 vertex  and merge its incident edges, replacing the two edges $(u_1, v)$ and $(v, u_2)$ (with lengths $w_1$ and $w_2$) with edge $(u_1, u_2)$ (with length $w_1 + w_2$). 

The projection $\Psi_{\mathcal{L}}$ can also be defined on split- and orthant-valued arguments. 

\begin{definition}
     \citep[Definition 2.4]{GrindstaffGillian2019RoPL} Given the split $s = \mathcal{G} \mid \mathcal{N}\setminus\mathcal{G}$ on $\mathcal{N}$, we define $\Psi_{\mathcal{L}}(s)$ to be the split on $\mathcal{L}$ that separates all leaves of $\mathcal{L}$ in $\mathcal{G}$ from those not in $\mathcal{G}$; that is, $ \Psi_{\mathcal{L}}(s) = \mathcal{L}\cap\mathcal{G} \mid \mathcal{L}\setminus \mathcal{G}$. 
    In a similar way, given a set of internal splits $S$ defining a topology orthant $\mathcal{O}(S)$, the projection onto $\mathcal{T}^{\mathcal{L}}$ is given by
\begin{equation*}
    \Psi_{\mathcal{L}}\left(\mathcal{O}(S)\right) = \mathcal{O}\left(\left\{\Psi_{\mathcal{L}}(s) | s \in S \text{ and }  \Psi_{\mathcal{L}}(s) \text{ is an internal split}\right\}\right).
\end{equation*}

\end{definition}

\begin{remark}
    It might be the case that either $\mathcal{G}\cap\mathcal{L} = \emptyset$ or $\mathcal{L} \setminus \mathcal{G} = \emptyset$, which means the projection $\Psi_{\mathcal{L}}(s)$ for $s = \mathcal{G} \mid \mathcal{N}\setminus \mathcal{G}$ is not a valid edge. This happens when this split is not involved in the shortest path between any pair of leaves contained in $\mathcal{L}$, which means this edge should be removed from the tree when applying $\Psi_{\mathcal{L}}$. In this case we say $\Psi_{\mathcal{L}}(s) = \emptyset$.
\end{remark}


\begin{definition}   \citep[Section 1.3]{ren2017combinatorial}
    For $\mathcal{L} \subseteq \mathcal{N}$ and a binary tree $T \in \mathcal{T}^{\mathcal{L}}$, the \textbf{connection cluster} $C_{T}^{\mathcal{N}}$ of $T$ in $\mathcal{T}^{\mathcal{N}}$ is the set of maximum-dimensional orthants in $\Psi_{\mathcal{L}}^{-1}(\mathcal{O}(T))$. 
\end{definition}

\begin{definition}
     \citep[Definition 3.9]{GrindstaffGillian2019RoPL} 
    Given a binary tree $T \in \mathcal{T}^{\mathcal{L}}$, with $\mathcal{L} \subseteq \mathcal{N}$, and an orthant $O \in C_{T}^{\mathcal{N}}$, let the \textbf{orthant-specific extension space}, $E_{T}^{O}$, be all trees in $O$ that are part of the extension space; that is, $E^{O}_{T} = E_{T}^{\mathcal{N}} \cap O$.  
\end{definition}

The previous two definitions are key in describing extension spaces, as the extension space can be described through the finite union $E_{T}^{\mathcal{N}} = \bigcup_{O \in C_T^{\mathcal{N}}} E_{T}^{O}$. Moreover, each $E_{T}^{O}$ is characterized  through a set of linear equations. Namely, given a binary tree $T \in \mathcal{T}^{\mathcal{L}}$ with $l$ leaves and an orthant $O \in C_{T}^{\mathcal{N}}$, consider the list of $2l - 3$ splits $q_1, \hdots, q_{2l-3}$ in $\mathcal{S}(T)$ and all $2n-3$ splits $p_1, \hdots, p_{2n-3}$ in  $\mathcal{S}(O)$. Note that for every split $p_j$, either $\Psi_{\mathcal{L}}(p_j) = q_i$ for some unique value $i \in  \{1, \hdots, 2l-3\}$, or $\Psi_{\mathcal{L}}(p_j) = \emptyset$. The TDR map acts linearly on the lengths of the splits because the removal of a leaf bisecting a branch results in a branch length that is the sum of the lengths of the adjacent branches. Formally, 
\begin{equation*}
    |q_i|_{\Psi_{\mathcal{L}}(T')} = \sum_{p_j: \Psi_{\mathcal{L}}(p_j) = q_i} |p_j|_{T'},
\end{equation*}
for all $T' \in O$. Thus, the TDR map for the orthant $O$ can be described by a linear map given by a $(2l-3) \times (2n-3)$-dimensional matrix, 
\begin{equation*}
    M_{\mathcal{L}}^{O}[i,j] = \begin{cases}
        1 & \text{ if } \Psi_{\mathcal{L}}(p_j) = q_i,\\
        0 & \text{ otherwise}.
    \end{cases}
\end{equation*}
Thus, $E^{O}_T$ can be described by the system of equations
\begin{equation}
\label{eq:EquationsOrthant}
    \begin{split}
        M_{\mathcal{L}}^{O} \mathbf{x} = \mathbf{v}_{T}, \\
        \mathbf{x} \geq 0,
    \end{split}
\end{equation}
where $\mathbf{x}$ is a $(2n-3)-$dimensional vector in orthant $O$ that gives the branch lengths of trees in $E_T^O$ and $\mathbf{v}_T$ is the vector representation of $T$ in $\mathcal{O}(T)$; that is, it is the $(2l-3)-$dimensional vector with the edge lengths of $T$.  

\begin{definition}
    The $j$-th column of the matrix $M_{\mathcal{L}}^{O}$ corresponds to the edge $p_j$ of the trees in $O$, which by definition will be a zero column if there is no edge $q_i$ to which $p_j$ maps into; i.e. $\Psi_{\mathcal{L}}(p_j) = \emptyset$. Thus, these edges as \textbf{inconsequential} for this extension space, since their length is unimportant for $\mathbf{x}$ to be a solution to \eqref{eq:EquationsOrthant}. If an edge is not inconsequential, it is \textbf{consequential}.
\end{definition}

Each column in $M_{\mathcal{L}}^{O}$ corresponding to a consequential edge has exactly one entry equal to 1 and the rest equal to 0, implying the rows are linearly independent. Moreover, every edge $q_i$ in $T$ has a non-empty pre-image under $\Psi_{\mathcal{L}}$, which means the matrix has no zero rows. Thus, $M_{\mathcal{L}}^{O}$ is of rank $2l - 3$. 

\begin{example} 
Consider again the trees in Fig.~3, and take $O$ to be the orthant of $T'$. The internal splits of trees in $O$ are: $s_1 = \{A,B\} \mid \{C,D,E,F,G\}$, $s_2 = \{C,D,E\} \mid \{A,B,F,G\}$, $s_3 = \{D,E\} \mid \{A,B,C,F,G\}$ and $s_4 = \{F,G\} \mid \{A,B,C,D,E\}$. Any tree in $E_T^O$ will solve \eqref{eq:EquationsOrthant} for the following values of $M_{\mathcal{L}}^{O}$ and $\mathbf{v}_T$. 
We have highlighted all zero columns in the projection matrix, and note that these coincide with the inconsequential edges shown in Fig.~3(b).
   \begin{align*}
      \mathbf{x}^{\top}  &
= \left(\begin{array}{ccccccccccc}
    e_A & e_B & e_C & e_D & e_E & e_F & e_G & s_1 & s_2 & s_3 & s_4
  \end{array}\right), \\
M_{\mathcal{L}}^{O} &
= \left(\begin{array}{cccc>{\columncolor{lightgray!20}}c>{\columncolor{lightgray!20}}c>{\columncolor{lightgray!20}}cccc>{\columncolor{lightgray!20}}c}
    1 & 0 & 0 & 0 & 0 & 0 & 0 & 0 & 0 & 0 & 0  \\
    0 & 1 & 0 & 0 & 0 & 0 & 0 & 0 & 0 & 0 & 0 \\
    0 & 0 & 1 & 0 & 0 & 0 & 0 & 0 & 0 & 0 & 0 \\
    0 & 0 & 0 & 1 & 0 & 0 & 0 & 0 & 0 & 1 & 0 \\
    0 & 0 & 0 & 0 & 0 & 0 & 0 & 1 & 1 & 0 & 0 
  \end{array}\right),
       \hspace{0.25cm}\mathbf{v}_{T} = \left(\begin{array}{ccccc}
  2 & 2 & 1 & 3 & 4
\end{array}\right), 
    \end{align*}
\end{example}




So far we have described extension spaces through the union of a finite number of subsets, each contained in a single maximum-dimensional orthant. These maximum-dimensional orthants are in themselves Euclidean spaces (Remark \ref{Remark:euclideanOrthants}) and the extension spaces restricted to maximum-dimensional orthants are affine subspaces. These observations lay the groundwork for the development of an algorithm to compute distances between orthant-specific extension spaces based on gradient of descent methods (Section \ref{sec:ESOdistance}). We provide an overview of gradient of descent methods in the following section.

\section{Reduced gradient methods} 
\label{sec:OptimizationOverview}

Gradient descent methods are a class of iterative algorithms that find stationary points of a function. At each iteration of a gradient descent algorithm, the local changes of the function around the current point are evaluated, and a ``step'' in best local direction is taken. 
In their simplest form, these methods are defined on a continuous differentiable function $f: D \mapsto \mathbb{R}$ on a domain $D \subset \mathbb{R}^{n}$. They rely on the idea that the gradient $\nabla f$ describes local behavior of $f$ and at each $x \in D$, and the vector $-\nabla f(x)$ provides a direction towards which the value of $f$ is locally decreasing the most. 
By starting at $x_0$ and repeatedly taking small steps in a decreasing direction, the algorithm will approach local minimum. While these algorithms are not guaranteed to find the global minimum of the function (or even to converge), certain conditions on $f$ and $D$ ensure good behavior. In this overview, we will focus on convex functions over convex and compact domains. In general, for convex differentiable functions we can guarantee the existence of a global minimum, and that this minimum will be reachable by gradient descent methods  \citep[Theorem 3.1, Theorem 5.4]{NLoptimization}. 

\label{sec:ReduxGDM}

The reduced gradient method is applicable when there are constraints on the domain of $f$. That is, we wish to minimize $f: \mathbb{R}^{n} \mapsto \mathbb{R}$ subject to two constraints: 
\begin{itemize}
    \item \textbf{Linear constraints: } $A \mathbf{x} = \mathbf{b}$ for some matrix $A \in \mathbb{R}^{m\times n}$ of rank $m$ and a fixed $\mathbf{b} \in \mathbb{R}^{m}$, and
    \item \textbf{Inequality constraints: } $\mathbf{x} \geq 0$. 
\end{itemize}
The set of values $\mathcal{X} \subset \mathbb{R}^{n}$ holding these two constraints is called the \textbf{feasible set}. Within this feasible set, the algorithm searches by moving along different facets, where a \textbf{facet} is a subset of $\mathcal{X}$ with certain entries fixed at zero. Specifically, a facet consists of all vectors $x \in \mathcal{X}$ such that $x_{i} = 0$ for each $i$ in some pre-specified subset of indices. 

At each step of the reduced gradient method, we determine the direction of change by considering the gradient of $f$, but we limit the possible directions of change based on the constraints. Since $A$ is a matrix of rank $m$, we know that the system of equations $A \mathbf{x} = \mathbf{b}$ has $n-m$ degrees of freedom, that is, we can consider $m$ entries of $\mathbf{x}$ as depending on the other $n-m$ free variables directly. 
If we consider the gradient of $f$ as a function of only the free variables, and determine the change on the dependent variables afterwards, we can find a new point that remains in the feasible set. 

In what follows, we use the notation from \cite[Section 6.1.2]{NLoptimization}. At each step of the reduced gradient method, we classify the variables in $\mathbf{x}$ into three groups:
\begin{itemize}
    \item Null variables $\mathbf{N}$: Variables fixed at zero. $x_i = 0$ when $i \in \mathbf{N}$.
    \item Free variables $\mathbf{F}$: Independent variables. $x_i$ can take any non-negative value when $i\in \mathbf{F}$.
    \item Dependent variables $\mathbf{D}$: Knowing the values of all other variables allows us to find $x_i$, $i \in \mathbf{D}$, to satisfy the linear constraints.  
\end{itemize}
We use the notation $\mathbf{x}_{\mathbf{D}}$ to refer to the vector containing only the entries of $\mathbf{x}$ corresponding to the dependent variables, $\mathbf{x}_{\mathbf{F}}$ the vector containing all the free variables, $\mathbf{x}_{\mathbf{N}}$ the vector containing all the null variables, and write 
    $f(\mathbf{x}) = f(\mathbf{x}_{\mathbf{D}}, \mathbf{x}_{\mathbf{F}}, \mathbf{x}_{\mathbf{N}})$
to make explicit the dependence of $f$ on the variables as the index sets vary. We also use $M_{\mathbf{I}}$ to refer to the matrix obtained by subseting the columns of $M$ indexed by $\mathbf{I}$.

The reduced gradient method begins by initializing the index sets, selecting $m$ variables to be the dependent variables in $\mathbf{D}^{0} \subset \{1,\hdots, n\}$ so that $A_{\mathbf{D}^{0}}$ is a square matrix of rank $m$, then taking $\mathbf{F}^{0} = \{1,\hdots, n\} \setminus \mathbf{D}^{0}$ and $\mathbf{N}^{0} = \emptyset$. An initial point $\mathbf{x}^{0}$ such that $A\mathbf{x}^{0} = \mathbf{b}$ and $\mathbf{x}^{0} > 0$ is also initialized. In the $k$-th iteration:

\begin{enumerate}
    \item Define the function $\varphi$ that takes the values of $f$ but depends only on the free variables. Null variables are fixed at zero and the dependent variables are those that guarantee the linear constraints, $\mathbf{x}_{\mathbf{D}^{k}} = A_{\mathbf{D}^{k}}^{-1} \left[\mathbf{b} - A_{\mathbf{F}^k} \mathbf{x}_{\mathbf{F}^{k}}\right]$. Therefore,
    $\varphi(\mathbf{x}_{\mathbf{F}^{k}}) = f\left(A_{\mathbf{D}^{k}}^{-1} \left[\mathbf{b} - A_{\mathbf{F}^k} \mathbf{x}_{\mathbf{F}^{k}}\right], \mathbf{x}_{\mathbf{F}^{k}}, 0\right).$
    \item Compute the gradient of this function by applying the chain rule:
    \begin{equation}
        \nabla \varphi(\mathbf{x}_{\mathbf{F}^{k}}) = \nabla_{\mathbf{F}^{k}} f(\mathbf{x}) - A_{\mathbf{F}^{k}}^{\top} \left[A_{\mathbf{D}^{k}}^{-1}\right]^{\top} \nabla_{\mathbf{D}^{k}} f(\mathbf{x}).
    \end{equation}
    \item Find an optimal direction of change $\mathbf{d}^{k} = (\mathbf{d}^{k}_{\mathbf{D}}, \mathbf{d}^{k}_{\mathbf{F}}, \mathbf{d}^{k}_{\mathbf{N}})$ in the feasible set by: 
    \begin{enumerate}[i.]
        \item Employing an unconstrained optimization method to find $\mathbf{d}^{k}_{\mathbf{F}}$ based on $\nabla \varphi(\mathbf{x}_{\mathbf{F}^{k}})$ (e.g., via the conjugate gradient method  \citep[Section 5.5]{NLoptimization}).
        \item Setting $\mathbf{d}^{k}_{\mathbf{D}} = -A_{\mathbf{D}}^{-1} A_{\mathbf{F}} \mathbf{d}^{k}_{\mathbf{F}}$ to ensure linear constraints. 
        \item Keeping the null variables fixed at zero, i.e., $\mathbf{d}^{k}_{\mathbf{N}} = 0$. 
    \end{enumerate}
    \item Find the value $\tau^{*}$ that minimizes
        $f(\mathbf{x}^{k} + \tau \mathbf{d}^{k}),$
    maintaining the inequality constraint, $\mathbf{x}^{k} + \tau^{*} \mathbf{d}^{k} \geq 0$ (e.g., using a line search method  \citep[Section 5.2]{NLoptimization}).
    \item Select the next point by taking $\mathbf{x}^{k+1} = \mathbf{x}^{k} + \tau^{*} \mathbf{d}^{k}$. 
    \item Verify if variables need to be reclassified. If $\langle \nabla f(\mathbf{x}^{k+1}), d^{k}\rangle < 0$, then $\tau^{*}$ 
    was selected for being the maximum value of $\tau$ holding the inequality constraint. Some non-null variable in $\mathbf{x}^{k+1}$ has value zero and needs to be reclassified: 
    \begin{enumerate}[i.]
        \item If $x^{k+1}_{i} = 0$ for some $i \in \mathbf{F}^{k}$, then reclassify $x_i$ as a null variable (if there is multiple options, select one), setting $\mathbf{F}^{k+1} = \mathbf{F}^{k} \setminus \{i\}$ and $\mathbf{N}^{k+1} = \mathbf{N}^{k} \cup \{i\}$.
        \item Otherwise, $x^{k+1}_{i} = 0$ for some $i \in \mathbf{D}^{k}$ (if there is multiple, select one). In this case we reclassify $x^{k+1}_{i}$ as a null variable, but a new dependent variable must be found to satisfy the linear constraint. We find $j \in \mathbf{F}^{k}$ such that setting $\mathbf{D}^{k+1} = \{j\}\cup \mathbf{D}^{k} \setminus \{i\} $ makes $A_{\mathbf{D}^{k+1}}$ still of rank $m$  \citep[Lemma 6.2]{NLoptimization}. Then, $\mathbf{F}^{k+1} = \mathbf{F}^{k}\setminus \{j\}$ and $\mathbf{N}^{k+1} = \mathbf{N}^{k} \cup \{i\}$. 
    \end{enumerate}
    \item Repeat 1-6 until either $\nabla \varphi(\mathbf{x}^{k}_{\mathbf{F}^{k}}) = 0$ or $\mathbf{F}^{k} = \emptyset$. 
    \item Check if the current point $\mathbf{x}^{k}$ is the global solution by computing 
                $\bar{\mu}^{k} = \nabla_{\mathbf{N}^{k}} f\left(\mathbf{x}^{k}\right) - A_{ \mathbf{N}^{k}}^{\top} \left\{A_{ \mathbf{D}^{k}}^{\top}\right\}^{-1}  \nabla_{\mathbf{D}^{k}} f\left(\mathbf{x}^{k}\right)$
            \cite[Lemma 6.3]{NLoptimization}. 
        If $\bar{\mu}^{k}_j \geq 0$ for every entry $j$, end the algorithm and return $\mathbf{x}^{k}$. 
        If $\bar{\mu}^{k}_{j} < 0$ for any entry $j$, increase the search region by taking $\mathbf{N}_0 \subset \mathbf{N}^{k}$ so that $\bar{\mu}^{k}_i < 0$ for all $i \in \mathbf{N}_0$ and reclassifying $\mathbf{F}^{k+1} = \mathbf{F}^{k} \cup \mathbf{N}_0$ and $\mathbf{N}^{k+1} = \mathbf{N}^{k} \setminus \mathbf{N}_0$, and return to 1-6. 
\end{enumerate}

If $f$ is convex and the feasible set $ \mathcal{X} \subset \mathbb{R}^{n}$ is compact, the number of stationary points that will be reached in Step 7 is finite, and therefore the global minimum will eventually be attained  \citep[Theorem 6.4]{NLoptimization}. 

\section{Distances between extension spaces}
\label{sec:ESOdistance}

Given two trees $T_1$ and $T_2$ with leaf sets $\mathcal{L}_1, \mathcal{L}_2 \subseteq \mathcal{N}$, we study computing the distance between their extension spaces $E^{\mathcal{N}}_{T_1}$ and $E^{\mathcal{N}}_{T_2}$ (see \eqref{eq:distanceDefinition}). 
If a pair of trees $(t^{*}_1,t^{*}_2) \in  E_{T_1}^{\mathcal{N}
} \times E_{T_2}^{\mathcal{N}
}$ is such that $d(t^{*}_1,t^{*}_2) =  d\left(E_{T_1}^{\mathcal{N}
}, E_{T_2}^{\mathcal{N}}\right)$, we call this an \textbf{optimal pair}. 
As we will see, an optimal pair is guaranteed to exist by the convexity properties of orthant-specific extension spaces and their finite products,  however, optimal pairs may not be unique. 

\subsection{Search region}
\label{subSect:OESproperties}

As discussed in Section \ref{sec:ExtensionSpace}, the extension space $E_{T}^{\mathcal{N}}$ for a tree $T \in \mathcal{T}^{\mathcal{L}}$ can be described as the finite union of subsets restricted to maximum-dimensional orthants, $\bigcup_{O \in C_{T}^{\mathcal{N}}} E^{O}_{T}$. Our algorithm to find the distance between extension spaces will find the minimal distance within a given orthant pair $(O_{1}, O_{2}) \in C_{T_1}^{\mathcal{N}} \times C_{T_2}^{\mathcal{N}}$. In this way, we construct a finite number of candidates for the optimal pair. We begin by demonstrating that we can exclude poor candidates for optimal pairs, and in this way, obtain convex and compact search regions.

\begin{definition}
    Given trees $T_1 \in \mathcal{T}^{\mathcal{L}_1}$ and $T_2 \in \mathcal{T}^{\mathcal{L}_2}$, where $\mathcal{L}_1, \mathcal{L}_2 \subseteq \mathcal{N}$, and an orthant pair $(O_1, O_2) \in  C_{T_1}^{\mathcal{N}} \times  C_{T_2}^{\mathcal{N}}$, we define the \textbf{orthant-specific mutually restricted extension space}, $\lceil{E^{O_1}_{T_1} \times E^{O_2}_{T_2}}\rceil$, as the subset of all tree pairs $(T'_1, T'_2) \in E^{\mathcal{O}_1}_{T_1} \times E^{\mathcal{O}_2}_{T_2}$ for which common edges that are inconsequential in either $T_1$ or $T_2$ are the same length in both trees, common edges that are inconsequential in both trees are length zero, and inconsequential uncommon edges are of length zero. That is, 
    $(T'_1, T'_2) \in E^{O_1}_{T_1} \times E^{O_2}_{T_2}$ is in $\lceil{E^{O_1}_{T_1} \times E^{O_2}_{T_2}}\rceil$ when the  \textbf{conditions for mutual restriction} hold: 
    \begin{itemize}
        \item $|p|_{T'_1} = |p|_{T'_2}$ for all $p \in \mathcal{P}(O_1) \cap \mathcal{P}(O_2)$ with $\Psi_{\mathcal{L}_1}(p) = \emptyset$ or $\Psi_{\mathcal{L}_2}(p) = \emptyset$,
        \item $|p|_{T'_1} = |p|_{T'_2} = 0$ for all $p \in \mathcal{P}(O_1) \cap \mathcal{P}(O_2)$ with $\Psi_{\mathcal{L}_1}(p) = \emptyset$ and $\Psi_{\mathcal{L}_2}(p) = \emptyset$,
        \item $|p|_{T'_1} = 0$ for every $p \in \mathcal{P}(O_1) \setminus \mathcal{P}(O_2)$ with $\Psi_{\mathcal{L}_1}(p) = \emptyset$,
        \item $|p|_{T'_2} = 0$ for every $p \in \mathcal{P}(O_2) \setminus \mathcal{P}(O_1)$ with $\Psi_{\mathcal{L}_2}(p) = \emptyset$.
\end{itemize}
\end{definition}

This definition is motivated by the length of any inconsequential edge $p$ in $T'_1\in O_1$ not affecting whether $T'_1$ is a part of the 
extension of $T_1$. However, the length of $|p|_{T'_1}$ influences the distance from this tree to trees in the orthant-specific extension space $E_{T_2}^{O_2}$. When building a tree $T_1'$ as a candidate for the optimal pair, once the lengths of consequential edges are chosen to ensure $T'_1 \in E_{T_1}^{O_1}$, we can freely select the length of $p$ to minimize the distance. 

Lemmas \ref{lemma:ProjectingCloser} and \ref{Lemma:distanceInconsequentialBranches} formalize this idea, showing that for any pair $(T'_1, T'_2) \in E^{O_1}_{T_1} \times E^{O_2}_{T_2}$, we can construct a new tree pair $(T^*_1, T^*_2) \in \lceil{E^{O_1}_{T_1} \times E^{O_2}_{T_2}}\rceil$ such that $d(T^*_1, T^*_2) \leq d(T'_1, T'_2)$. Or first result concerns uncommon inconsequential edges. 


\begin{lemma}
    \label{lemma:ProjectingCloser}
    Consider two trees $T_1, T_2 \in \mathcal{T}^{\mathcal{N}}$ with at least one uncommon edge $p \in \mathcal{P}(T_1) \setminus \mathcal{P}(T_2)$ in $T_1$. Define $T^{\perp p}_1$ to be the projection of $T_1$ onto the face of $\mathcal{O}(T_1)$ defined by the length of $p$ being equal to zero; that is, $T^{\perp p}_1$ is such that $T^{\perp p}_1 \in \mathcal{O}(T_1)$ with $|p|_{T^{\perp p}_1} = 0$ and $|p'|_{T^{\perp p}_1} = |p'|_{T_1}$ for all $p' \neq p$. Then, 
    $d(T^{\perp p}_1, T_2) \leq d(T_1, T_2)$.
\end{lemma}

\begin{proof}
    Consider the support $(\mathcal{A}, \mathcal{B})$ of the path space of the geodesic from $T_1$ to $T_2$.  For $i = 1, \hdots, k$, we denote by $t_i \in \mathcal{T}^{\mathcal{N}}$ the first tree on the geodesic from $T_1$ to $T_2$ that belongs to the orthant
    $O_i = \mathcal{O}(C \cup B_1 \cup \hdots \cup B_{i} \cup A_{i+1} \cup \hdots  A_{k})$
    in the path space of the geodesic. This means $t_i$ is a tree in the face shared by $O_{i-1}$ and $O_i$. Since $p$ is an edge of $T_1$ and not of $T_2$, then there is a value $r \in {1, \hdots, k}$ such that $p \in A_{r}$. We will construct a shorter path from $T^{\perp p}_1$ to $t_r$ than the geodesic from $T_1$ to $t_r$ by projecting every $t_i$ for $i \leq r$ towards the face of the respective orthant where the length of the edge $p$ is zero. 
    
    For every $i < r$ we have $p \in S(O_i)$. Similarly as before, for each $i \leq r$, define $t^{\perp p}_i$ to be the tree in $O_i$ such that $|p|_{t^{\perp p}_i} = 0$ and $|p'|_{t_i^{\perp p}} = |p'|_{t_i}$ for every other edge $p' \in \mathcal{P}(O_{i})$ such that $p'\neq p$. In particular, since $t_r$ is a tree in the face shared by $O_{r-1}$ and $O_r$, and $p\in A_r$ implies $p \notin \mathcal{P}(O_r)$, then $|p|_{t_r} = 0$ and $t_r^{\perp p} = t_r$. 

    We denote $t_0 = T_1$. For every $i=1, \hdots, r$, both $t_{i-1}$ and $t_i$ are in the orthant $O_{i-1}$, so 
    $d(t_{i-1}, t_i) = \sqrt{ \sum_{p' \in \mathcal{P}(O_{i-1})} \left(|p'|_{t_{i-1}} - |p'|_{t_i}\right)^2}$. Similarly, 
    $d(t^{\perp p}_{i-1}, t^{\perp p}_i) = \sqrt{\sum_{p' \in \mathcal{P}(O_{i-1})}  \left(|p'|_{t^{\perp p}_{i-1}} - |p'|_{t^{\perp p}_i}\right)^2}$. In the latter expression the difference of lengths for edge $p$ will be zero and all other differences will be as in the former expression, therefore $d(t^{\perp p}_{i-1}, t^{\perp p}_i) \leq d(t_{i-1}, t_i)$. Thus  $d(T_1, T_2) \geq d(T^{\perp p}_{0}, t^{\perp p}_1) + d(t^{\perp p}_{1}, t^{\perp p}_2) + \hdots + d(t^{\perp p}_{r-1}, t_r) + d(t_r, T_2) \geq d(T^{\perp p}_1, T_2)$
    %
\end{proof}

We now give a stronger result which demonstrates that we can find minimal distances between extension spaces by searching orthant-specific mutually restricted extension spaces. 


\begin{lemma}
    \label{Lemma:distanceInconsequentialBranches}
    Given two trees $T_1$ and $T_2$ with leaf sets $\mathcal{L}_1, \mathcal{L}_2 \subseteq \mathcal{N}$ and an orthant pair $(O_1, O_2) \in  C_{T_1}^{\mathcal{N}} \times  C_{T_2}^{\mathcal{N}}$, the distance between the orthant-specific extension spaces equals the distance between their orthant-specific mutually restricted extension space, that is, 
    \begin{equation*}
        \inf \left\{d(T'_1, T'_2) \left| (T'_1, T'_2) \in E^{O_1}_{T_1} \times E^{O_2}_{T_2}\right. \right\} = \inf \left\{d(T'_1, T'_2) \left|(T'_1, T'_2) \in  \lceil{E^{O_1}_{T_1} \times E^{O_2}_{T_2}}\rceil \right. \right\}.
    \end{equation*}
\end{lemma}

\begin{proof}
     Since $\lceil{E^{O_1}_{T_1} \times E^{O_2}_{T_2}}\rceil \subseteq E^{O_1}_{T_1} \times E^{O_2}_{T_2}$, we need only show that 
    for any pair $(T'_1, T'_2) \in E^{O_1}_{T_1} \times E^{O_2}_{T_2}$ there is a pair $(T^{*}_1, T^{*}_2) \in \lceil{E^{O_1}_{T_1} \times E^{O_2}_{T_2}}\rceil$ such that $d(T^{*}_1, T^{*}_2) \leq d (T'_1, T'_2)$. We will construct such a pair by first setting all inconsequential common edges to a proper length and all uncommon inconsequential edges to length zero. 

    We begin by defining $(T^0_1, T^0_2)$ such that $\mathcal{O}(T^0_1) \subseteq O_1$ and $\mathcal{O}(T^0_2) \subseteq O_2$, $|s|_{T^{0}_1} = |s|_{T'_1}$ for all $s\in \mathcal{P}(O_1)\setminus \mathcal{P}(O_2)$, $|s|_{T^{0}_2} = |s|_{T'_2}$ for all $s\in \mathcal{P}(O_2)\setminus \mathcal{P}(O_1)$, and the conditions for mutual restriction hold for common edges $p \in \mathcal{P}(O_1) \cap \mathcal{P}(O_2)$:
    \begin{equation*}
        \begin{split}
             |p|_{T^{0}_1} = |p|_{T'_1}, |p|_{T^{0}_2} = |p|_{T'_2} &\text{ when } \Psi_{\mathcal{L}_1}(p) \neq \emptyset \text { and } \Psi_{\mathcal{L}_2}(p) \neq \emptyset\\
            |p|_{T^{0}_1} = |p|_{T^{0}_2} = |p|_{T'_2} &\text{ when } \Psi_{\mathcal{L}_1}(p) = \emptyset \text{ and } \Psi_{\mathcal{L}_2}(p) \neq \emptyset \\
            |p|_{T^{0}_2} = |p|_{T^{0}_1} = |p|_{T'_1} &\text{ when } \Psi_{\mathcal{L}_1}(p) \neq \emptyset \text{ and } \Psi_{\mathcal{L}_2}(p) = \emptyset \\
            |p|_{T^{0}_1} = |p|_{T^{0}_2} = 0 &\text{ when } \Psi_{\mathcal{L}_1}(p) = \emptyset \text{ and } \Psi_{\mathcal{L}_2}(p) = \emptyset \\
        \end{split}
    \end{equation*}
     Since the lengths of the uncommon internal splits are unchanged from $(T'_1, T'_2)$ to $(T^{0}_1, T^{0}_2)$, 
    the support of the path space will be the same for both pairs (common edges lengths do not influence the support for the geodesic  \citep[Section 4]{OwenMegan2011}). For those edges that are inconsequential in one or both trees, the difference in length drops to zero; it remains the same for edges that are consequential for both (for such an edge $p$, $|p|_{T'_1} - |p|_{T'_2} =|p|_{T^{0}_1} - |p|_{T^{0}_2}$). We denote by $K^{0}$ the common edges (including external edges) that are consequential for both extension spaces (i.e. $s \in K^{0} \subseteq K = C \cup H$ when $\Psi_{\mathcal{L}_1}(s) \neq \emptyset$ and $\Psi_{\mathcal{L}_2}(s) \neq \emptyset$) 
    \begin{equation*}
        \begin{split}
            d(T^{0}_1, T^{0}_2) &= \left|\left| \left(||A_1||_{T^{0}_1} + ||B_1||_{T^{0}_2}, \hdots, ||A_k||_{T^{0}_1} + ||B_k||_{T^{0}_2}, \left(|s|_{T^{0}_1} - |s|_{T^{0}_2}\right)_{s \in K} \right) \right|\right|\\ 
            &= \left|\left| \left(||A_1||_{T^{0}_1} + ||B_1||_{T^{0}_2}, \hdots, ||A_k||_{T^{0}_1} + ||B_k||_{T^{0}_2}, \left(|s|_{T^{0}_1} - |s|_{T^{0}_2}\right)_{s \in K^0}, 0\right) \right|\right|\\
            &= \left|\left| \left(||A_1||_{T'_1} + ||B_1||_{T'_2}, \hdots, ||A_k||_{T'_1} + ||B_k||_{T'_2}, \left(|s|_{T'_1} - |s|_{T'_2}\right)_{s \in K^0}, 0\right) \right|\right|\\
            &\leq d(T'_1, T'_2).
        \end{split}
    \end{equation*}
    Although all common edges in the tree pair $(T^0_1, T^0_2)$ hold the conditions for mutual restrictions, some of the uncommon edges between both trees may not. Define the inconsequential edges in $T^0_1$ that are not common with $T^0_2$ by $(p^1_1, \hdots, p^1_{r_1})$, and the inconsequential edges in $T^0_2$ not common with $T^0_1$ by $(p^2_1, \hdots, p^2_{r_2})$. For each $i=1, 2$, define $T^{j}_i = T^{(j-1) \perp p^i_{j}}_i$ the projection of $T^{j-1}_i$ towards the face of $O_i$ defined by the length of $p^i_j$ being equal to zero. By repeatedly applying Lemma \ref{lemma:ProjectingCloser}, we have 
    \begin{equation*}
        d(T^0_1, T^0_2) \geq d(T^1_1, T^0_2) \geq \hdots \geq d(T^{r_1}_1, T^0_2) \geq d(T^{r_1}_1, T^1_2) \geq \hdots \geq d(T^{r_1}_1, T^{r_2}_2).
    \end{equation*}
    Thus, $d(T^{r_1}_1, T^{r_2}_2) \leq d(T'_1, T'_2)$, and by construction, $(T^{r_1}_1, T^{r_2}_2) \in \lceil{E^{O_1}_{T_1} \times E^{O_2}_{T_2}}\rceil$. 
\end{proof}

Having identified a subspace for each orthant pair that contains a minimum-distance tree pair, we can now guarantee algorithmic convergence. Our goal is to minimize $d: \mathcal{T}^{\mathcal{N}} \times \mathcal{T}^{\mathcal{N}} \mapsto \mathbb{R}_{\geq 0}$. Since $\mathcal{T}^{N}$ is of non-positive curvature, the function $d$ is doubly convex  \citep[Definition 1.9, Corollary 2.5]{sturm}; i.e., $d$ is a convex function on $\mathcal{T}^{\mathcal{N}} \times \mathcal{T}^{\mathcal{N}}$, which is in itself also a geodesic space  \citep[Proposition 5.3]{Bridson:1999ky}.
We now show that the search region $\lceil{E^{O_1}_{T_1} \times E^{O_2}_{T_2}}\rceil$ is convex and compact, which will ensure the convergence of our proposed gradient descent method (\cite[Theorem 6.4]{NLoptimization}). 

\begin{lemma}
    \label{lemma:convexityInOrthant}
    For any two trees $T_1$ and $T_2$ with leaf sets $\mathcal{L}_1, \mathcal{L}_2 \subseteq \mathcal{N}$ and an orthant pair $(O_1, O_2) \in  C_{T_1}^{\mathcal{N}} \times  C_{T_2}^{\mathcal{N}}$, the orthant-specific mutually restricted extension space $\lceil{E^{O_1}_{T_1} \times E^{O_2}_{T_2}}\rceil$ is a convex, closed and bounded subspace of $\mathcal{T}^{\mathcal{N}}\times \mathcal{T}^{\mathcal{N}}$. 
\end{lemma}

\begin{proof}

     $O_1 \times O_2$ is an Euclidean subspace of the geodesic space $\mathcal{T}^{\mathcal{N}} \times \mathcal{T}^{\mathcal{N}}$ (Remark \ref{Remark:euclideanOrthants}), and $E^{O_1}_{T_1}$ and $E^{O_2}_{T_2}$ can both be described through a set of linear equations \eqref{eq:EquationsOrthant}, which implies that $E^{O_1}_{T_1} \times E^{O_2}_{T_2}$ is closed and convex. For each $p \in \mathcal{P}(O_1) \cup \mathcal{P}(O_2)$, define $\lceil O_1 \times O_2 \rceil^{p}$ to be the subset of trees that hold the condition for mutual restriction that applies to $p$. Each such subset can be expressed by a system of linear equations, which means each is a closed and convex subset. Thus, we can write the mutually restricted extension space as the intersection of a finite number of closed and convex subspaces,
     \begin{equation*}
         \lceil{E^{O_1}_{T_1} \times E^{O_2}_{T_2}}\rceil = E^{O_1}_{T_1} \times E^{O_2}_{T_2} \cap \left\{\bigcap_{\substack{p \in \mathcal{P}(O_1) \cup \mathcal{P}(O_2)\\ \Psi_{\mathcal{L}_1}(p) = \emptyset \text{ or } \Psi_{\mathcal{L}_2}(p) = \emptyset}} \lceil O_1 \times O_2 \rceil^{p}\right\},
     \end{equation*}
     which is therefore a convex and closed subspace. 
     
     Finally, consider a tree pair $(T'_1, T'_2) \in  \lceil{E^{O_1}_{T_1} \times E^{O_2}_{T_2}}\rceil$ and consider $|p|_{T'_1}$ for any edge $p \in \mathcal{P}(O_1)$. If $p$ is consequential, there is an edge $q \in \mathcal{P}(T_1)$ such that $q = \Psi_{\mathcal{L}_1}(p)$ and we know that $|q|_{T_1} = \sum_{p' | q = \Psi_{\mathcal{L}_1}(p)} |p'|_{T'_1}$, so that $0 \leq |p|_{T'_1} \leq |q|_{T_1}$. If $p$ is inconsequential, then either $|p|_{T'_1} = 0$, or $|p|_{T'_1} = |p|_{T'_2}$ where $p$ is consequential for $O_2$ (implying $0 \leq |p|_{T'_1} = |p|_{T'_2} \leq |q|_{T_2}$ for $q \in \mathcal{P}(T_2)$ such that $q = \Psi_{\mathcal{L}_2}(p)$). Thus, all edges in $T'_1$ are bounded, and likewise for $T'_2$. Therefore, $\lceil{E^{O_1}_{T_1} \times E^{O_2}_{T_2}}\rceil$ is bounded. 
\end{proof}

\subsection{Distances as a reduced gradient problem}
\label{subSect:Distance2ReducedGradient}

Having established properties of the search region, we now formulate the search for an optimal pair as a reduced gradient problem. Given $T_1$ and $T_2$ with leaf sets $\mathcal{L}_1, \mathcal{L}_2 \subseteq \mathcal{N}$ and orthants $O_1 \in C^{\mathcal{N}}_{T_1}$, $O_2 \in C^{\mathcal{N}}_{T_2}$, we wish to find $(T^{*}_1, T^{*}_2) \in \lceil E^{O_1}_{T_1} \times  E^{O_2}_{T_1}\rceil$ such that $d(T^{*}_1,T^{*}_2) \leq d(T'_1, T'_2)$ for all $(T'_1, T'_2) \in E^{O_1}_{T_1} \times E^{O_2}_{T_2}$. Consider the projection matrices $M_{\mathcal{L}_1}^{O_1}$ and $M_{\mathcal{L}_2}^{O_2}$ and fixed vectors $\mathbf{v}_{T_1}$ and $\mathbf{v}_{T_2}$ that describe the orthant-specific extension spaces $E^{O_1}_{T_1}$ and $E^{O_2}_{T_2}$. That is, the edge-lengths vectors $\mathbf{x}_{T'_1}$ and $\mathbf{x}_{T'_2}$ corresponding to any pair $(T'_1, T'_2) \in \lceil E^{O_1}_{T_1}\times E^{O_2}_{T_2}\rceil$ will be such that $M_{\mathcal{L}_1}^{O_1} \mathbf{x}_{T'_1} = \mathbf{v}_{T_1} \text{ and }  M_{\mathcal{L}_2}^{O_2} \mathbf{x}_{T'_2} = \mathbf{v}_{T_2}$. By Lemma \ref{Lemma:distanceInconsequentialBranches}, some of the elements of $\mathbf{x}_{T'_i}$ ($i=1,2$) equal each other or equal zero. These values correspond to all the inconsequential edges, i.e., the zero columns in the projection matrices $M_{\mathcal{L}_1}^{O_1}$ and $M_{\mathcal{L}_2}^{O_2}$. 
Define the reduced matrices $\ddot{M}^{i} = \left[M_{\mathcal{L}_i}^{O_i}\right]_{\mathbf{R_i}}$, where $\mathbf{R_i}$ gives the indices of non-zero columns in $M_{\mathcal{L}_i}^{O_i}$. Similarly, we define  $\ddot{\mathbf{x}}_{T'_i} := [\mathbf{x}_{T'_i}]_{\mathbf{R_i}}$ by subsetting to consequential edges in $\mathbf{x}_{T'_i}$. Note that since the only columns eliminated from the projection matrices are zero vectors, the system of linear equations defined by
$M_{T_i}^{O_i} \mathbf{x}_{T'_i} = \mathbf{v}_i$
is equivalent to $\ddot{M}^{i} \ddot{\mathbf{x}}_{T'_i} = \mathbf{v}_i$. Therefore, defining
\begin{equation}
    \dot{\mathbf{M}} = \begin{pmatrix}
        \ddot{M}^{1} & 0 \\
        0 & \ddot{M}^{2}
    \end{pmatrix} \text{, } \dot{\mathbf{x}} = \begin{pmatrix}
        \ddot{\mathbf{x}}_{T'_1}\\
        \ddot{\mathbf{x}}_{T'_2}
    \end{pmatrix} \text{ and } \dot{\mathbf{v}} = \begin{pmatrix}
        \mathbf{v}_1 \\
        \mathbf{v}_2
    \end{pmatrix},
\end{equation}
we can describe $\lceil E^{O_1}_{T_1}\times E^{O_2}_{T_2}\rceil$ through the system of linear equations
\begin{equation}
\label{Eq:CompleteSystemEquation}
    \dot{\mathbf{M}} \dot{\mathbf{x}} = \dot{\mathbf{v}} \text{ with } \dot{\mathbf{x}} \geq 0.
\end{equation}

Given a solution to \eqref{Eq:CompleteSystemEquation}, we can construct trees $(T'_1, T'_2) \in \lceil E^{O_1}_{T_1} \times  E^{O_2}_{T_1} \rceil$ by assigning the values in $\dot{\mathbf{x}}$ corresponding to edges in $O_i$ to the lengths of the edges in $T'_i$, letting $|p|_{T_1} = |p|_{T_2}$ for all common edges $p$ that are inconsequential for one of the trees, $|p|_{T_1} = |p|_{T_2} = 0$ for common edges inconsequential in both trees and $|p|_{T_i} = 0$ for uncommon edges that are inconsequential in the respective tree. Note that in this way, all values of edges in $T'_1$ and $T'_2$ are unambiguously defined. We let $(T'_1(\dot{\mathbf{x}})$, $T'_2(\dot{\mathbf{x}}))$ refer to the unique pair of trees constructed in this manner from a solution vector $\dot{\mathbf{x}}$. We summarise the above in the following result.

\begin{algorithm}
\small
\caption{A reduced gradient method to find BHV distances between orthant-specific extension spaces (see additional details on Page 23)}
\label{alg:OrthantExtensionDistance}
\begin{algorithmic}
\STATE{Set initial values: $\dot{\mathbf{x}}^0_j = \sum_i \mathbf{1}_{\{\dot{\mathbf{M}}[i,j] > 0\}} \times \dot{\mathbf{v}}_i/\#\{j: \dot{\mathbf{M}}[i,j] = 1\}$. }
\STATE{Define initial index sets: 
For each $i$, add the index $j$ of the first column such that $\dot{\mathbf{M}}[i,j] = 1$ to \textbf{D}. Add all $j'>j$ such that $\dot{\mathbf{M}}[i,j'] = 1$ to \textbf{F}. Set \textbf{N}$= \emptyset$. }
\STATE{Initialize $c_{\text{conj}} = 1$. Set tolerance thresholds $\texttt{Tol1}$ and $\texttt{Tol2}$.}
\WHILE{global minimum not reached}
\STATE{Compute gradient 
 $\nabla \varphi(\dot{\mathbf{x}}^{t}_{\textbf{F}}) = \nabla_{\textbf{F}} \delta(\dot{\mathbf{x}}^{t}) -  \dot{\mathbf{M}}_{\textbf{F}}^{\top} \dot{\mathbf{M}}_{\textbf{D}}^{-\top}\nabla_{\textbf{D}} \delta(\dot{\mathbf{x}}^{t})  $}
\IF{$||\nabla \varphi(\dot{\mathbf{x}}^{t}_{\textbf{F}})||_\infty < \texttt{Tol1}$ } 
\STATE{Compute $\overline{g}_{\textbf{N}} = \nabla_{\textbf{N}} \delta(\dot{\mathbf{x}}^{t}) - \dot{\mathbf{M}}_{\textbf{N}}^{\top}\dot{\mathbf{M}}_{\textbf{D}}^{-\top} \nabla_{\textbf{D}} \delta(\dot{\mathbf{x}}^{t})$.} 
\IF{$\overline{g}_{\textbf{N}} \geq 0$}
\STATE{\textbf{stop while:} global minimum has been reached.}
\ELSE
\STATE{Define $\textbf{N}_{p} = \{j \in \textbf{N} \mid \overline{g}_{\textbf{N}}[j]<0\}$}
\STATE{Update $\textbf{F} = \textbf{F} \cup \textbf{N}_{p}$ and $\textbf{N} = \textbf{N}\setminus \textbf{N}_{p}$}
\ENDIF
\ENDIF
\STATE{Compute $\mathbf{d}^{t}_{\textbf{F}} = - \nabla \varphi(\dot{\mathbf{x}}^{t}_{\textbf{F}}) + \mathbf{1}_{\{c_{\text{conj}}=1\}}\frac{\langle\nabla \varphi(\dot{\mathbf{x}}^t_{\textbf{F}}), \nabla \varphi(\dot{\mathbf{x}}^t_{\textbf{F}}) - \nabla \varphi(\dot{\mathbf{x}}^{t-1}_{\textbf{F}})\rangle}{||\nabla \varphi(\dot{\mathbf{x}}^{t-1}_{\textbf{F}})||^2} \mathbf{d}^{t-1}_{\textbf{F}} $}
\STATE{Set $\mathbf{d}^t_{\textbf{D}} = - \dot{\mathbf{M}}_{\textbf{D}}^{-1} \dot{\mathbf{M}}_{ \textbf{F}} \mathbf{d}^t_{\textbf{F}}$ and $\mathbf{d}^t_{\textbf{N}} = 0$. Increase $c_{\text{conj}} \leftarrow c_{\text{conj}} + 1$}
\STATE{Find $\tau_{\text{max}} = \max\{\tau \geq 0 \mid \dot{\mathbf{x}}^t + \tau \mathbf{d}^{t} \geq 0 \}$ and $h(\tau_{\text{max}}) = \langle \nabla \delta(\dot{\mathbf{x}}^{t} + \tau_{\text{max}} \mathbf{d}^{t}), \mathbf{d}^{t} \rangle$}
\IF{$h(\tau_{\text{max}}) \leq 0$}
\STATE {Set $\tau_{0} = \tau_{\text{max}}$}
\ELSE
\STATE{Set $\tau_{\text{left}} = 0$, $\tau_{\text{right}} = \tau_{\text{max}}$ and $\tau_{0} = \frac{\tau_{\text{left}} + \tau_{\text{right}}}{2}$}
\WHILE{$|h(\tau_{0}) = \langle \nabla \delta(\dot{\mathbf{x}}^{t} + \tau_{0} \mathbf{d}^{t}), \mathbf{d}^{t} \rangle| > \texttt{Tol2}$}
\STATE{ \textbf{if} $h(\tau^{*}) > 0$ \textbf{then} set $\tau_{\text{right}} = \tau^{*}$ and $\tau^{*} = \frac{\tau_{\text{left}} + \tau_{\text{right}}}{2}$}
\STATE{\textbf{else} set $\tau_{\text{left}} = \tau^{*}$ and $\tau^{*} = \frac{\tau_{\text{left}} + \tau_{\text{right}}}{2}$.}
\ENDWHILE
\ENDIF
\IF{$\dot{\mathbf{x}}^t_j + \tau_0d^t_j = 0$ for some $j \in \textbf{D} \cup \textbf{F}$}
\STATE{Select $j \in \textbf{D}\cup \textbf{F}$ such that $\dot{\mathbf{x}}^t_j + \tau_0d^t_j = 0$}
\IF{$j \in \textbf{F}$}
\STATE{Set $\textbf{F} = \textbf{F}\setminus\{j\}$ and $\textbf{N} = \textbf{N} \cup \{j\}$}
\ELSE
\STATE{Select $j'\in \textbf{F}$ such that $\dot{\mathbf{M}}[i,j] = \dot{\mathbf{M}}[i,j'] = 1$ for some index $i$}
\STATE{Set $\textbf{D} = \textbf{D}\setminus (\{j\} \cup \{j'\})$, $\textbf{F} = \textbf{F}\setminus \{j'\}$ and $\textbf{N} = \textbf{N}\cup\{j\}$}
\ENDIF
\ENDIF
\STATE{Update $\dot{\mathbf{x}}^{t+1} = \dot{\mathbf{x}}^{t} + \tau^{*} \mathbf{d}_{t}$}
\STATE{\textbf{if} $c_{\text{conj}} + 1 > 15$ \textbf{then} $c_{\text{conj}}  = 1$ \textbf{else} $c_{\text{conj}} \leftarrow c_{\text{conj}} + 1 $}
\ENDWHILE
\RETURN {$(T'_1(\dot{\mathbf{x}}^t), T'_2(\dot{\mathbf{x}}^t))$ and $\sqrt{\delta(T'_1(\dot{\mathbf{x}}^t),T'_2(\dot{\mathbf{x}}^t))}$ }
\end{algorithmic}
\end{algorithm}

\begin{lemma}    \label{lemma:EquivalentMinimizingProblem}
    Given $\dot{\mathbf{M}}$ and $\dot{\mathbf{v}}$ for trees $T_1 \in \mathcal{T}^{\mathcal{L}_1}$ and $T_2 \in \mathcal{T}^{\mathcal{L}_2}$, consider 
    \begin{equation}
        \label{eq:RealMinimizingProblem}
            \dot{\mathbf{x}}^{*} \in  \underset{\dot{\mathbf{M}} \dot{\mathbf{x}} = \dot{\mathbf{v}}, \dot{\mathbf{x}} \geq 0}{\arg\min} d\left(T'_1(\dot{\mathbf{x}}), T'_2(\dot{\mathbf{x}})\right).
    \end{equation}
    Then $d(T'_1(\dot{\mathbf{x}}^{*}), T'_2(\dot{\mathbf{x}}^{*})) \leq d(T'_1,T'_2)$ for any $(T'_1,T'_2) \in O_1 \times O_2$. 
    
\end{lemma}

\begin{proof}
    Take $\mathbf{U} = \left\{\dot{\mathbf{x}} \geq 0 | \dot{\mathbf{M}} \dot{\mathbf{x}} = \dot{\mathbf{v}} \right\}$. The function $\chi: \mathbf{U} \mapsto \lceil{E^{O_1}_{T_1} \times E^{O_2}_{T_2}}\rceil$ that maps $\chi(\dot{\mathbf{x}}) = (T'_1(\dot{\mathbf{x}}^{*}), T'_2(\dot{\mathbf{x}}^{*}))$ is bijective:
    \begin{enumerate}
        \item If $\dot{\mathbf{x}}_1 \neq \dot{\mathbf{x}}_2$ are two different solutions to \eqref{eq:EquationsOrthant}, then at least one consequential edge in $T'_1(\dot{\mathbf{x}}_1)$ or $T'_2(\dot{\mathbf{x}}_1)$ has a different length than the same consequential edge in $T'_1(\dot{\mathbf{x}}_2)$ or $T'_2(\dot{\mathbf{x}}_2)$. Thus, $\chi$ is injective.
        \item Given a pair of trees $(T'_1, T'_2) \in \lceil{E^{O_1}_{T_1} \times E^{O_2}_{T_2}}\rceil$, construct the vector $\dot{\mathbf{x}}$ by subsetting $\mathbf{x}_{T'_1}$ and $\mathbf{x}_{T'_2}$ to only consequential edges. Since these trees are in the extension spaces, $\dot{\mathbf{x}}$ would be a solution to \eqref{Eq:CompleteSystemEquation}, so $\chi$ is surjective. 
    \end{enumerate}
    Thus, finding the minimum distance pair in $\lceil{E^{O_1}_{T_1} \times E^{O_2}_{T_2}}\rceil$ is equivalent to solving \eqref{eq:RealMinimizingProblem}. The result follows from Lemma \ref{Lemma:distanceInconsequentialBranches}.
    
\end{proof}


\subsection{Objective function and gradient}
\label{subSect:GradientsBHV}

While we are interested in minimizing $d(\cdot, \cdot)$, in practice we solve the equivalent problem of minimizing $d^2$, 
\begin{equation}
    \label{def:gamma2}
    d^2\left(T'_1, T'_2\right) = \sum_{i = 1}^{k} \left(||A_{i}||_{T'_1} + ||B_{i}||_{T'_2}\right)^2 + \sum_{s \in K} \left(|s|_{T'_1} - |s|_{T'_2}\right)^2.
\end{equation}
Define $\delta\left(\dot{\mathbf{x}}\right) := d^2\left(T'_1\left(\dot{\mathbf{x}}\right), T'_2\left(\dot{\mathbf{x}}\right)\right)$. We denote by $\dot{x}^j_{p}$ the entry in $\dot{\mathbf{x}}$ corresponding to the edge $p \in \bigcup_{i = 1}^{k} \{A_i \cup B_i\} \cup K$ 
in $T'_j$. By construction of $\dot{\mathbf{x}}$, $\dot{x}^j_{p}$ is only well-defined if $p$ is a consequential edge for $T_j$.
For any subset of edges $S \subseteq \mathcal{P}(T'_j)$, denote by $\dot{S}$ the edges in $S$ that are consequential for $T_j$, and define $||\dot{S}|| = \sqrt{\sum_{p \in \dot{S}} \left(\dot{x}^{j}_p\right)^2}$. Given that any inconsequential uncommon edge is of length zero for trees in $\lceil{E^{O_1}_{T_1} \times E^{O_2}_{T_2}}\rceil$, then $||A_i||_{T'_1(\dot{\mathbf{x}})} = ||\dot{A}_i||$ and $||B_i||_{T'_2(\dot{\mathbf{x}})} = ||\dot{B}_i||$ for every $i = 1, \hdots, k$. We also use $\dot{K}$ to refer to the set of common splits that are consequential for both trees
Since any squared term in the last sum in \eqref{def:gamma2} involving an inconsequential edge is zero as well, we write the function $\delta$ as
\begin{equation}
    \label{equation:ObjectiveFunction}
    \delta\left(\dot{\mathbf{x}}\right) = \sum_{i = 1}^{k} \left(||\dot{A_{i}}|| + ||\dot{B_{i}}||\right)^2 + \sum_{s \in \dot{K}} \left(\dot{x}^1_{s} -\dot{x}^2_{s}\right)^2 
\end{equation}
Note the support $(\mathcal{A},\mathcal{B})$ will depend on $\dot{\mathbf{x}}$ indirectly through the trees $T'_1(\dot{\mathbf{x}})$ and $T'_2(\dot{\mathbf{x}})$, but the following lemma ensures that this dependence does not affect the continuity and convexity of the function. Additionally, the gradient method we employ requires $\delta(\dot{\mathbf{x}})$ to be continuously differentiable, which we also address.

\begin{lemma}
    \label{lemma:ObjFunctionContConv}
    The function $\delta: \mathbb{R}_{\geq 0}^{\mathbf{r_1} + \mathbf{r_2}} \mapsto \mathbb{R}_{\geq 0}$, where $\mathbf{r}_j$ is the number of non-zero columns in $M^{O_j}_{\mathcal{L}_j}$, is a continuous and convex function. Moreover, $\delta$ is continuously differentiable in the interior of the domain $\mathbb{R}_{> 0}^{\mathbf{r_1} + \mathbf{r_2}}$. 
\end{lemma}

\begin{proof}
    Consider the map $\chi: \mathbb{R}_{\geq 0}^{\mathbf{r_1} + \mathbf{r_2}} \mapsto O_1 \times O_2$ given by $\chi(\dot{\mathbf{x}}) = \left(T'_1(\dot{\mathbf{x}}), T'_2(\dot{\mathbf{x}})\right)$. As discussed previously, $O_1 \times O_2$ is a $(4n-6)-$dimensional Euclidean space, and by definition, each of the $4n-6$ coordinates in the image of $\chi(\dot{\mathbf{x}})$ is either one of the values of $\dot{\mathbf{x}}$ (the value corresponding to the consequential edge) or fixed to zero. Thus, $\chi$ is a linear mapping between Euclidean spaces. And since the function $d^2: O_1 \times O_2 \mapsto  \mathbb{R}_{\geq 0}$ is continuous and convex, then $\delta = d^2 \circ \chi$ is also continuous and convex  \citep[Page 79]{alma99124373890001452}. 

    Note each variable $\dot{x}^{j}_{p}$ in \eqref{equation:ObjectiveFunction} contributes to exactly one quadratic term. Thus, the gradient of $\delta: \mathbb{R}_{\geq 0}^{\mathbf{r_1} + \mathbf{r_2}} \mapsto \mathbb{R}_{\geq 0}$ has entries  
    \begin{equation}
    \label{equation:PartialDerivativesExpression}
        \frac{\partial \delta(\dot{\mathbf{x}})}{\partial \dot{x}^{j}_{p}} = \begin{cases}
            2\dot{x}^{j}_{p}\left(1 + \frac{||\dot{B_{i}}||}{||\dot{A_{i}}||}\right) & j=1, p \in \dot{A}_{i}\\
            2\left(\dot{x}^1_{p} -\dot{x}^2_{p}\right) & j = 1, p \in \dot{K} \\
            2\dot{x}^{j}_{p}\left( 1 + \frac{||\dot{A_{i}}||}{||\dot{B_{i}}||} \right) & j=2, p \in \dot{B}_{i}\\
             2\left(\dot{x}^2_{p} -\dot{x}^1_{p}\right) & j=2, p \in \dot{K}.\\
        \end{cases}
    \end{equation}
    Since we are using the unique minimal support of the geodesic between trees $T'_1(\dot{\mathbf{x}})$ and $T'_2(\dot{\mathbf{x}})$ in \eqref{equation:ObjectiveFunction} and \eqref{equation:PartialDerivativesExpression}, and these trees are uniquely and well-defined by $\dot{\mathbf{x}}$, the partial derivatives given by \eqref{equation:PartialDerivativesExpression} are well-defined as long as $||\dot{A_i}||, ||\dot{B_i}|| \neq 0$, which will be the case within the domain's interior. Other support $(\mathcal{A}',\mathcal{B}')$ for the geodesic from $T'_1(\dot{\mathbf{x}})$ to $T'_2(\dot{\mathbf{x}})$ will hold the property 
    $$\frac{||\dot{B'_{l}}||}{||\dot{A'_{l}}||} = \frac{||B'_{l}||_{T'_2(\dot{\mathbf{x}})}}{||A'_{l}||_{T'_1(\dot{\mathbf{x}})}} =  \frac{||B_{i}||_{T'_2(\dot{\mathbf{x}})}}{||A_{i}||_{T'_1(\dot{\mathbf{x}})}} = \frac{||\dot{B_{i}}||}{||\dot{A_{i}}||} \text{ when } p \in A_{i} \text{ and } p \in A'_{l}, $$ 
    $$\frac{||\dot{A'_{l}}||}{||\dot{B'_{l}}||} = \frac{||A'_{l}||_{T'_1(\dot{\mathbf{x}})}}{||B'_{l}||_{T'_2(\dot{\mathbf{x}})}} =  \frac{||A_{i}||_{T'_1(\dot{\mathbf{x}})}}{||B_{i}||_{T'_2(\dot{\mathbf{x}})}} = \frac{||\dot{A_{i}}||}{||\dot{B_{i}}||} \text{ when } p \in B_{i} \text{ and } p \in B'_{l}, $$
    which means the partial derivative with respect to any $\dot{x}^{i}_{p}$ would be equal under equivalent supports, which implies it is unambiguously defined. Since the map $\dot{\mathbf{x}} \mapsto \frac{||A||_{T'_1(\dot{\mathbf{x}})}}{||B||_{T'_2(\dot{\mathbf{x}})}}$ for any nonempty subsets $A \subseteq \mathcal{S}(T'_1)$ and $B \subseteq \mathcal{S}(T'_2)$ is continuous in the interior of the domain, the partial derivatives are continuous as well. 
\end{proof}



The continuous differentiability of $\delta$ extends beyond the interior of the domain. The gradient remains continuous at boundary points of the domain 
where \eqref{equation:PartialDerivativesExpression} is well-defined, i.e., where $||\dot{A_i}|| \neq 0 $ and $ ||\dot{B_i}|| \neq 0$. However, the gradient does not exist at points where $\dot{x}^1_{p} = 0$ for all $p \in \dot{A}_{i}$ or $\dot{x}^2_{p} = 0$ for all $p\in \dot{B}_{i}$. 
In these cases we can replace the non-existent gradient with a subgradient  \citep[Definition 2.72]{NLoptimization} without impacting the convergence of the algorithm. 
Specifically, in place of $\nabla \varphi(\dot{\mathbf{x}}^{k}_{\mathbf{F}^{k}})$ in the pausing condition $\nabla \varphi(\dot{\mathbf{x}}^{k}_{\mathbf{F}^{k}}) = 0$ (Step 7, Section \ref{sec:ReduxGDM}), we use a subgradient, 
replacing 
every undefined gradient entry $\frac{\partial \delta(\dot{\mathbf{x}})}{\partial \dot{x}^{j}_{p}}$ with zero. 
Note that the new (sub)gradient vector will equal zero when zero belongs to the set of subgradients $\partial \delta(\dot{\mathbf{x}}^{k})$, and that a sufficient condition for attaining a minimum of a convex function is for zero to belong to this set  \citep[Theorem 3.5]{NLoptimization}. The optimality condition for the global minimum (Step 8, Section \ref{sec:ReduxGDM})
can be replaced by an equivalent condition in which we require the existence of $\eta \in \partial \delta(\mathbf{x}^{k})$ 
such that 
$\eta_{\mathbf{N}^{k}} - \dot{\mathbf{M}}_{ \mathbf{N}^{k}}^{\top} \left\{\dot{\mathbf{M}}_{ \mathbf{D}^{k}}^{\top}\right\}^{-1}  \eta_{\mathbf{D}^{k}} \geq 0$  \citep[Theorem 3.34]{NLoptimization}.

\subsection{Algorithm for distances between extension spaces}
\label{subSect:AlgorithmFinal}

We are now able to describe our algorithm to solve \eqref{eq:distanceDefinition} and find an optimal pair. We begin with a result on using \ref{alg:OrthantExtensionDistance} to identify the closest trees within \textit{orthant-specific} extension spaces. 

\begin{theorem}
\label{Theorem1}
    Algorithm \ref{alg:OrthantExtensionDistance} converges to trees $(T^{*}_1,T^{*}_2) \in (E^{O_1}_{T_1}, E^{O_2}_{T_2})$ such that $$d(T^{*}_1,T^{*}_2) = \inf_{(t_1,t_2) \in (E^{O_1}_{T_1}, E^{O_2}_{T_2})} d(t_1, t_2).$$
\end{theorem}

\begin{proof}
    Algorithm \ref{alg:OrthantExtensionDistance} is a reduced gradient method to minimize $\delta: \mathbb{R}_{\geq 0}^{\mathbf{r_1} + \mathbf{r_2}} \mapsto \mathbb{R}_{\geq 0}$ 
    subject to constraints $\dot{\mathbf{M}} \dot{\mathbf{x}} = \dot{\mathbf{v}}$ and $\dot{\mathbf{x}}\geq 0$. 
    This function is continuous and convex (Lemma \ref{lemma:ObjFunctionContConv}). 
    Consider the feasible set $\mathbf{U} = \left\{\dot{\mathbf{x}} \geq 0 | \dot{\mathbf{M}} \dot{\mathbf{x}} = \dot{\mathbf{v}} \right\}$. Given that the function $\chi: \mathbf{U} \mapsto \lceil{E^{O_1}_{T_1} \times E^{O_2}_{T_2}}\rceil$ that maps $\chi(\dot{\mathbf{x}}) = (T'_1(\dot{\mathbf{x}}), T'_2(\dot{\mathbf{x}}))$ is a continuous bijective linear map, and  $\lceil{E^{O_1}_{T_1} \times E^{O_2}_{T_2}}\rceil$ is a convex, closed and bounded set, then $\mathbf{U}$ is as well. It follows that the algorithm converges to a point minimizing $\delta$ inside the feasible set (\cite[Theorem 6.4]{NLoptimization}). 
    Because of the bijection between $\mathbf{U}$ and $\lceil{E^{O_1}_{T_1} \times E^{O_2}_{T_2}}\rceil$, we conclude that the algorithm's solution $\dot{\mathbf{x}}^{*}$ 
    minimizes $d^2$ (and by extension $d$) on $\lceil{E^{O_1}_{T_1} \times E^{O_2}_{T_2}}\rceil$. The result follows from Lemma \ref{Lemma:distanceInconsequentialBranches}.
\end{proof}

\begin{theorem}
\label{Theorem2}
    For each orthant pair $(O_1, O_2) \in C^\mathcal{N}_{T_1} \times C^\mathcal{N}_{T_2}$, apply Algorithm \ref{alg:OrthantExtensionDistance} to construct a candidate pair $(T^{*}_1,T^{*}_2)$ for the optimal pair. Among all these candidate pairs, the one with the minimum distance will be a solution to \eqref{eq:distanceDefinition}.
\end{theorem}

\begin{proof}
Since the number of orthants in each connection cluster is finite, we can list every orthant in it $C^\mathcal{N}_{T_i} = \{O_{i}^{1},\hdots, O_i^{c_i}\}$. Denote by $(T^{*}_{1,j_1},T^{*}_{2,j_2})$ the tree pair obtained by applying Algorithm \ref{alg:OrthantExtensionDistance} to $(O_{1}^{j_1}, O_{2}^{j_2})$, $j_1 = 1, \ldots, c_1$, $j_2 = 1, \ldots, c_2$. $\{d(T^{*}_{1,j_1},T^{*}_{2,j_2})\}$ is a finite set, and we define $(T^{*}_{1},T^{*}_{2})$ to be the pair achieving its minimum. 
For any $(T'_1, T'_2) \in E_{T_1}^{\mathcal{N}}\times E_{T_2}^{\mathcal{N}}$, 
$\mathcal{O}(T_1) \subseteq O_1^{j_1}$ and $\mathcal{O}(T_2) \subseteq O_1^{j_2}$ for some $j_1$ and $j_2$. Therefore, $d(T^{*}_{1},T^{*}_{2}) \leq d(T^{*}_{1,j_1},T^{*}_{2,j_2}) \leq d(T'_1, T'_2)$
\end{proof}

\section{Algorithmic complexity and runtime}
\label{sec:Experiments}

Having described our method in Theorem \ref{Theorem2}, we now turn our attention to studying its performance 
as the shared and total number of leaves varies. 
We implemented our algorithm in Java (version 20.0.2). Our implementation is available as part of the \texttt{BHVExtMinDistance} library, which can be accessed freely from the \texttt{ExtnSpaces} repository at \url{https://github.com/statdivlab/ExtnSpaces.git}.
Our implementation depends on the 
\texttt{distanceAlg1} and \texttt{polyAlg} libraries, available at 
the \texttt{BHVtreespace} github repository: \url{https://github.com/megan-owen/BHVtreespace.git}. 
Code and instructions to reproduce the following two sections' analysis are available at \url{https://github.com/statdivlab/ExtnSpaces_supplementary.git}.
To our knowledge, no other algorithms exist to find distances between extension spaces, and therefore, there are no methods to benchmark against. 

Since our algorithm performs an optimization routine for each orthant pair, the total number of orthant pairs is a major driver of the complexity of our algorithm. The number of orthants in the extension space in $\mathcal{T}^{\mathcal{N}}$ ($|\mathcal{N}| = n$) of a tree with $|\mathcal{L}_i| = l_i$ leaves is $(2n - 5)!! / (2l_i - 5)!!$ 
 \citep[Theorem 2.1]{ren2017combinatorial}, and therefore, the number of orthant pairs we must consider is $\Omega := \{(2n - 5)!!\}^2/\{(2l_1 - 5)!! \times (2l_2 - 5)!!\}$. This value has the potential to be considerably large, since the growth rate of the value $(2n-5)!!$ is super-exponential. Using Stirling's approximation (for large $l_1, l_2$ and $n$) we obtain 
$$\Omega \sim \left(2e^{-1}\right)^{2n-l_1-l_2} (n-2)^{2n-4}(l_1-2)^{2-l_1}(l_2-2)^{2-l_2}.$$
and therefore, $\Omega = O(n^{2n - (l_1 + l_2)})$ when $n - l_1$ and $n - l_2$ are constant.

\begin{table}[htbp]
    \centering
    \caption{The runtime of Algorithm \ref{alg:OrthantExtensionDistance} in practice. For each setting $\mathcal{S}$, we report the number of pairs of orthants to search over ($\Omega$); the total runtime for computing distances between $E^{\mathcal{N}}_{T_1}$ and $E^{\mathcal{N}}_{T_2}$ (min:sec); the number of iterations for each reduced gradient method to converge (mean [median, 90\% quantile and maximum]); the distance between $E^{\mathcal{N}}_{T_1}$ and $E^{\mathcal{N}}_{T_2}$; and the number of optimal pairs (``$\#$ pairs''). We observe that the number of orthant pairs is the largest factor contributing to runtime. 
    } 
    \label{tab:simulation_study}
    \resizebox{\textwidth}{!}{%
    \begin{tabular}{|c|c|c|c|c|r|c|r|r|}
        \hline
        $\mathcal{S}$ & $|\mathcal{L}_1 \cup \mathcal{L}_2|$ & $|\mathcal{L}_1|$ & $|\mathcal{L}_2|$ & $\Omega$ & Runtime & Iterations & $d(E_{T_1}^{\mathcal{N}}, E_{T_2}^{\mathcal{N}})$ & $\#$ pairs \\
        \hline
        \multicolumn{9}{|c|}{Unimodal distribution for Edge Lengths}\\
        \hline
        a & 7 & 6 & 4 & 2835 & 00:05 & 5.50 [5, 7, 174] & 6.675 & 1 \\
        \hline
        b & 7 & 6 & 4 & 2835 & 00:03 & 4.94 [5, 7, 22] & 4.378 & 1 \\
        \hline
        c & 7 & 5 & 4 & 19845 & 00:34 & 6.41 [6, 10, 64] & 0.268 & 1 \\
        \hline
        d & 10 & 9 & 8 & 2925 & 00:10 & 3.96 [4, 6, 17] & 18.497 & 1 \\
        \hline
        e & 10 & 8 & 8 & 38025 & 02:27 & 4.98 [5, 8, 46] & 15.710 & 1 \\
        \hline
        f & 10  & 8 & 7 & 418275 & 50:17 & 7.33 [6, 12, 68] & 9.459 & 1 \\
        \hline
        \multicolumn{9}{|c|}{Bimodal distribution for Edge Lengths}\\
        \hline
        a & 7 & 6 & 4 & 2835 & 00:04 & 4.56 [4, 7, 28] & 108.284 & 6 \\
        \hline
        b & 7 & 6 & 4 & 2835 & 00:02 & 4.60 [5, 6, 18] & 108.667 & 1 \\
        \hline
        c & 7 & 5 & 4 & 19845 & 00:29 & 6.00 [6, 9, 36] & 24.880 & 1 \\
        \hline
        d & 10 & 9 & 8 & 2925 & 00:08 & 4.03 [4, 6, 23] & 132.690 & 1 \\
        \hline
        e & 10 & 8 & 8 & 38025 & 02:27 & 4.97 [5, 8, 40] & 104.722 & 2 \\
        \hline
        f & 10 & 8 & 7 & 418275 & 47:45 & 7.53 [6, 12, 91] & 123.323 & 2\\
        \hline
    \end{tabular}%
    }
\end{table}

We study the in-practice scalability and performance of our algorithm using simulated phylogenetic trees. We selected 6 pairs of topologies 
(available at \url{https://github.com/statdivlab/ExtnSpaces_supplementary.git})
each with a different combination of $\mathcal{L}_1 \cup \mathcal{L}_2$, $\mathcal{L}_1$ and $\mathcal{L}_2$, and considered $\mathcal{N} = \mathcal{L}_1 \cup \mathcal{L}_2$. 
We considered two distributions for the edge lengths, resulting in 12 total simulation settings. 
The first edge length distribution is a lognormal distribution (mean $=5$, variance $=1$), reflecting a low-variance edge length scenario.
The second distribution is a mixture of two lognormal distributions. The first component has a mean of 5 and a variance of 1 (sampled with 75\% probability), and the second component has a mean of 60 and a variance of 10 (sampled with 25\% probability). This second, high-variance distribution reflects long-branch scenarios that commonly arise in practice. 
 
\label{subsec:Results1}

The results of our exploration can be found in Table 1. 
Running clock-times are based on an 8-core Apple M1 processor with 16GB of RAM. These times reflect the process of computing the distances for each orthant pair sequentially, and applying a merge sort algorithm. Our library supports multi-threading, allowing distances in different orthant pairs to be computed concurrently to reduce run-times, but we report single-thread times here for transparency. For this simulation, we selected a tolerance for $\nabla \varphi(\dot{\mathbf{x}}_{\mathbf{F}})$ of $10^{-8}$. 

As our algorithm performs an optimization process per each orthant pair, we expected the runtime to be approximately proportionally to the number of orthant pairs for a fixed search space dimension, which we generally find to be the case. For example, when the number of orthant pairs increased by a factor of 7 (from 2835 pairs in (a) and (b) to 19845 in (c)), the increase in runtime was approximately 7-fold, from 2-5 seconds to around 29-34 seconds. 
Similarly, when the number of orthant pairs increased by 11-fold (from (e) to (f)), runtimes increased by $\sim 11$-fold, and 
when the number of orthant pairs increased by 13-fold (from (d) to (e)), runtimes increased by $\sim 16$-fold. 


The number of orthant pairs appears to affect the runtime through two avenues: directly (via the number of reduced gradient methods to be performed), and indirectly (due to an increase in the average number of iterations required for convergence). The number of iterations to convergence is directly influenced by the number of leaves being attached to each tree to create their extension space, as each reduced gradient problem involves the linear constraints $\dot{\mathbf{M}} \dot{\mathbf{x}} = \dot{\mathbf{v}}$ with $(r_1+r_2) - \{2(l_1+l_2)-6\}$ (the difference between the number of consequential edges and the number of original edges in $T_1$ and $T_2$) degrees of freedom, which is upper-bounded by the number of leaves being added to the trees. Consequentially, the number of free variables in each iteration of our optimization mechanism is at most $2n - l_1 - l_2$. Convergence will take longer when this number is higher.

Unsurprisingly, we find that distances between extension spaces tend to be higher when edge lengths are heavy-tailed (bimodal distribution). Interestingly, while in all unimodal distribution cases, only one optimal pair was found, it was common to find more than one optimum under the bimodal distribution. Both of these observations can be explained by how BHV distances depend on edge lengths. 
As described in Section \ref{sec:OwenProvan}, when going from a tree $t_1$ to $t_2$, the common edges are present in the topologies of all trees on the geodesic, while uncommon edges present in $t_1$ are gradually swapped for uncommon edges in $t_2$. Intuitively, the size of the uncommon edges in $t_1$ and $t_2$ gradually change between zero and their original size. Similarly, the lengths of common edges gradually change from their lengths in $t_1$ to their lengths in $t_2$. Thus, BHV distances are increased by longer uncommon edges, and by common edges with significantly different lengths. The bimodal distribution allows for longer uncommon edges, and introduces more variability in the sizes of the common edges, which explains why the distances are higher. 

Another effect of edge lengths on BHV distances is that, in practice, if an edge in a tree is decidedly longer than the others, topology orthants in the connection cluster that involve attaching new edges to this edge tend to produce shorter geodesics. 
If this particularly long edge is such that edges mapping into it (under the tree dimensionality reduction map) are uncommon edges, then the large value of the length must reduce to zero at some point along the geodesic. Thus, dividing this long edge into several smaller edges through attaching edges into it will reduce the size of the geodesic. 
This also explains why more than one optimal pair was found in some of the cases where edge lengths were assigned through the bimodal distribution. If one of the edges in one of the trees is long, then the best candidates for the optimal pair arise from those attaching leaves along that edge, and the same geodesic length can be achieved by attaching the edges in the same places along the long edge, but in a different order. For example, $T_1$ (case (a), bimodal) has a long external edge incident to the leaf \textbf{L06}, and all 6 optimal pairs are such that $T_2^*$ has edges \textbf{L03}, \textbf{L04} and \textbf{L05} attached to that edge. Thus, we obtain 6 optimal pairs because there 6 different ways we can order these three external edges across the long edge.

\section{Application to prokaryotic gene trees}
\label{sec:RealDataAnalysis}

Here, we illustrate our method on gene trees spanning phylogenetically diverse prokaryotic lineages. 
Prokaryotes (bacteria and archaea) display a high degree of discordance in the genes they carry, with fewer than $\sim$1\% of a given organism's genes distributed ``universally'' across all bacteria \citep{dagan2006tree}. Thus, the comparison of two prokaryotic gene phylogenies will almost always require tools that can handle non-identical leaf sets, motivating the development of our method. 


We analyze gene trees from \cite{zhu2019phylogenomics}, focusing on two genes involved in essential tasks: cell division and repair. Specifically, we consider $T_1$ to be the gene tree for \textit{ftsA} (coding for a protein involved in cell division) and $T_2$ to be the gene tree for \textit{dinB} (coding for a DNA polymerase protein involved in translesion repair). We restrict our analysis to 10 phylogenetically diverse organisms spanning 2 domains of life; our complete leaf set $\mathcal{N}$ is given in Table 2. These organisms are found in diverse habitats, including the human gut, oral cavity, and tumors; as well as groundwater, treated water, and deep-sea hydrothermal vents. Out of ten total organisms, only five organisms have both genes, with 3 and 2 unshared genes carried by \textit{ftsA} and \textit{dinB}, respectively. Note that these genes could be truly absent, or they could be unobserved due to imperfections in genome reconstruction from metagenomes \citep{duarte2020sequencing,zaheer2018impact,royalty2019theoretical}.

\begin{table}[htbp]
    \centering
    \label{tab:Leaves}
    \caption{The complete leaf set $\mathcal{N}$ for the \textit{ftsA} and \textit{dinB} gene trees. }
\begin{tabular}{|l|l|c|c|}
\hline
Species                                  & Domain & \textit{ftsA} tree & \textit{dinB} tree \\ \hline
\textit{Actinomyces odontolyticus}                & Bacteria     & No                   & Yes                  \\
\textit{Fusobacterium nucleatum}                  & Bacteria     & Yes                  & Yes                  \\
\textit{Pseudomonas pelagia}                      & Bacteria     & Yes                  & Yes                  \\
\textit{Bacteroides fragilis}                     & Bacteria     & Yes                  & Yes                  \\
\textit{Candidatus Saccharibacteria TM7x}              & Bacteria     & Yes                  & No                   \\
\textit{Sphingomonas hengshuiensis}               & Bacteria     & Yes                  & Yes                  \\
\textit{Parcubacteria SG8-24}                  & Bacteria     & Yes                  & No                   \\
\textit{Vibrio scophthalmi}                       & Bacteria     & Yes                  & Yes                  \\
\textit{Candidatus Lokiarchaeota CR4}  & Archea       & No                   & Yes                  \\
\textit{Candidatus Odinarchaeota LCB4} & Archea       & No                   & Yes                  \\ \hline
\end{tabular}
\end{table}

$T_1$ and $T_2$ are shown in Fig.~4(a). While $E_{T_1}^\mathcal{N}$ spans 2145 orthants, and $E_{T_2}^\mathcal{N}$ spans 195 orthants, none of these orthants are shared between the two extension spaces. As a result, the compatibility measures of \cite{GrindstaffGillian2019RoPL} are not defined for these two trees. 
In contrast, our distance $d(E_{T_1}^\mathcal{N}, E_{T_2}^\mathcal{N})$ is always defined. We applied Algorithm \ref{alg:OrthantExtensionDistance} to every orthant pair in $C^\mathcal{N}_{T_1} \times C^\mathcal{N}_{T_2}$ and found that the distance between the extension spaces is 4.234, and that this value was attained as the distance between the trees in $\mathcal{T}^\mathcal{N}$ shown in Fig.~4(b). To search through 418275 pairs of orthants in $\mathcal{T}^\mathcal{N}$ took 21 minutes on a 8-core Apple M1 processor with 16GB of RAM in multi-threaded setting with a thread pool of size 8. 

In addition, our approach to computing distances between trees with non-identical leaf sets naturally lends itself to a ``supertree'' method for combining information from two trees. Specifically, because the minimum distance between extension spaces is the size of geodesic paths in $\mathcal{T}^{\mathcal{N}}$ the tree in $\mathcal{T}^{\mathcal{N}}$ that is the midpoint on such a geodesic averages the information across the true trees \citep{brown2020mean, MillerOwenProvan}. However, because minimum distance paths between extension spaces are not necessarily unique, this supertree measure may not be unique. This is in contrast with Fr\'echet means of trees that are all contained within the same BHV space, whose Fr\'echet means are uniquely defined \citep{sturm}. 

Two paths were minimum distance between the extension spaces of the \textit{ftsA} and \textit{dinB} trees. Similar to the cases discussed in Section \ref{subsec:Results1}, the two trees ($T_A$ and $T_{A'}$ in Fig.~4(b)) in the extension space for \textit{ftsA} are both produced by attaching new edges (corresponding to external edges to \textit{Ca. Lokiarchaeota CR4} and \textit{Ca. Odinarchaeota LCB4}) to a particularly long edge (the external edge to \textit{S. hengshuiensis}) in the same locations but in a different order --- a phenomenon discussed in Section \ref{sec:Experiments}. The internal edges resulting from attaching these edges (with lengths 0.72 and 1.30) are in both cases uncommon to the tree $T_{B} \in E_{T_2}^\mathcal{N}$, and thus these edges reduce to zero in length and are then dropped. Because these edges have the same lengths in $T_A$ and $T_{A'}$ and all other edges are equal length, the tree along the geodesics where the last of these two edges (the edge separating \textit{Ca. Lokiarchaeota CR4}, \textit{Ca. Odinarchaeota LCB4} and \textit{S. hengshuiensis} from all other organisms) is dropped is the same, and then the geodesic follows the same path to $T_B$. 
Although theoretically the existence of two optimal pairs could admit two supertrees, in this case the length of the geodesic section from the starting tree (either $T_A$ or $T'_A$) to the tree where both geodesics coincide is less than half the length of the total geodesic. Because of this, the mid-point is unchanged between $T_A$ and $T'_A$. This unique supertree is shown in Fig.~4(c). If the length of the section where both geodesics do not coincide were longer than half the length of the geodesics, then the midpoints would differ from each other. Nevertheless they would still share many edges in common and with the same length, and uncommon edges would have a counterpart with the same length. Thus, we expect non-unique supertrees to be similar in general.


\begin{figure}
    \label{fig:AllTrees304317}
    \centering
    \subfigure[]{\includegraphics[width=0.4\textwidth]{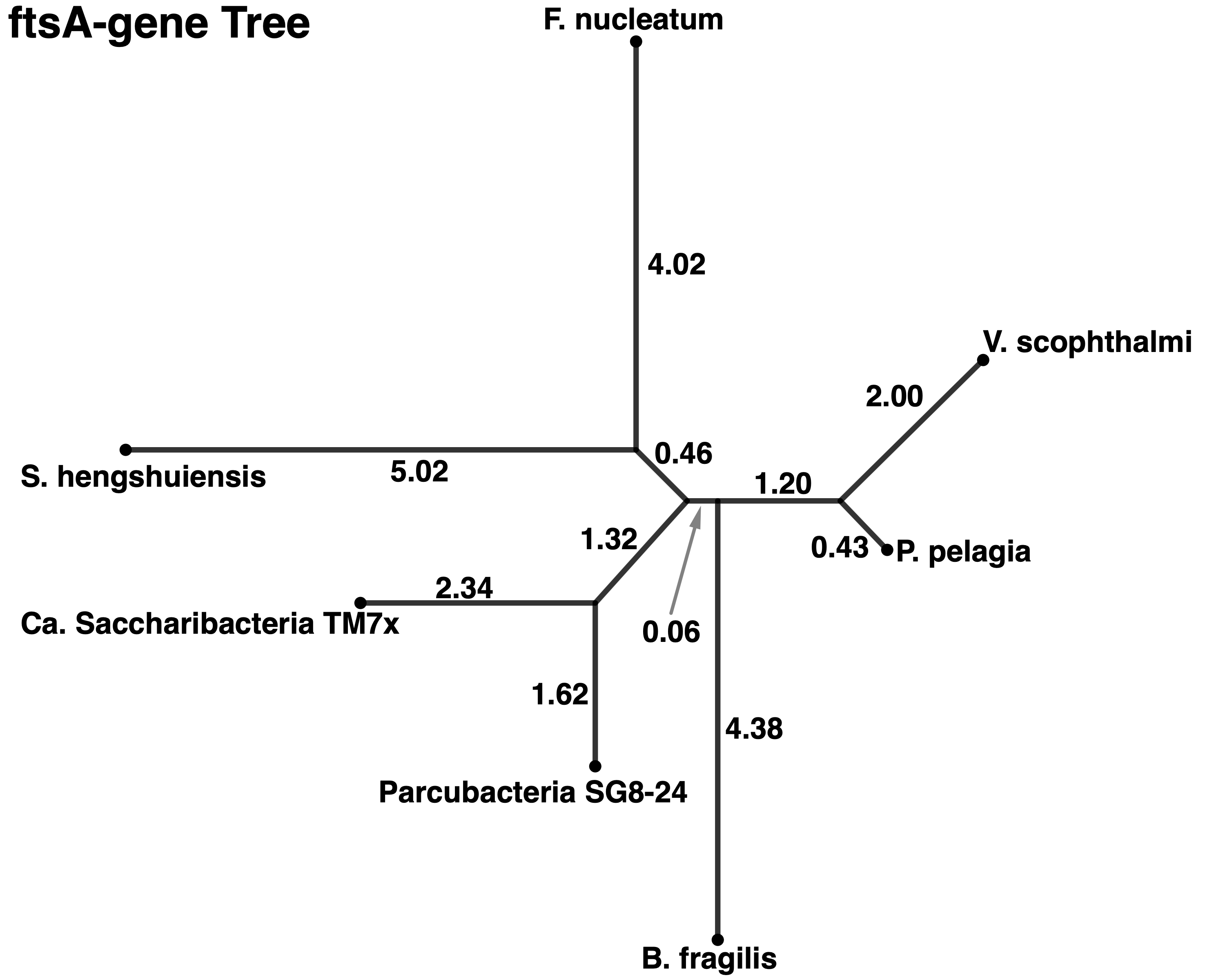}
    \includegraphics[width=0.4\textwidth]{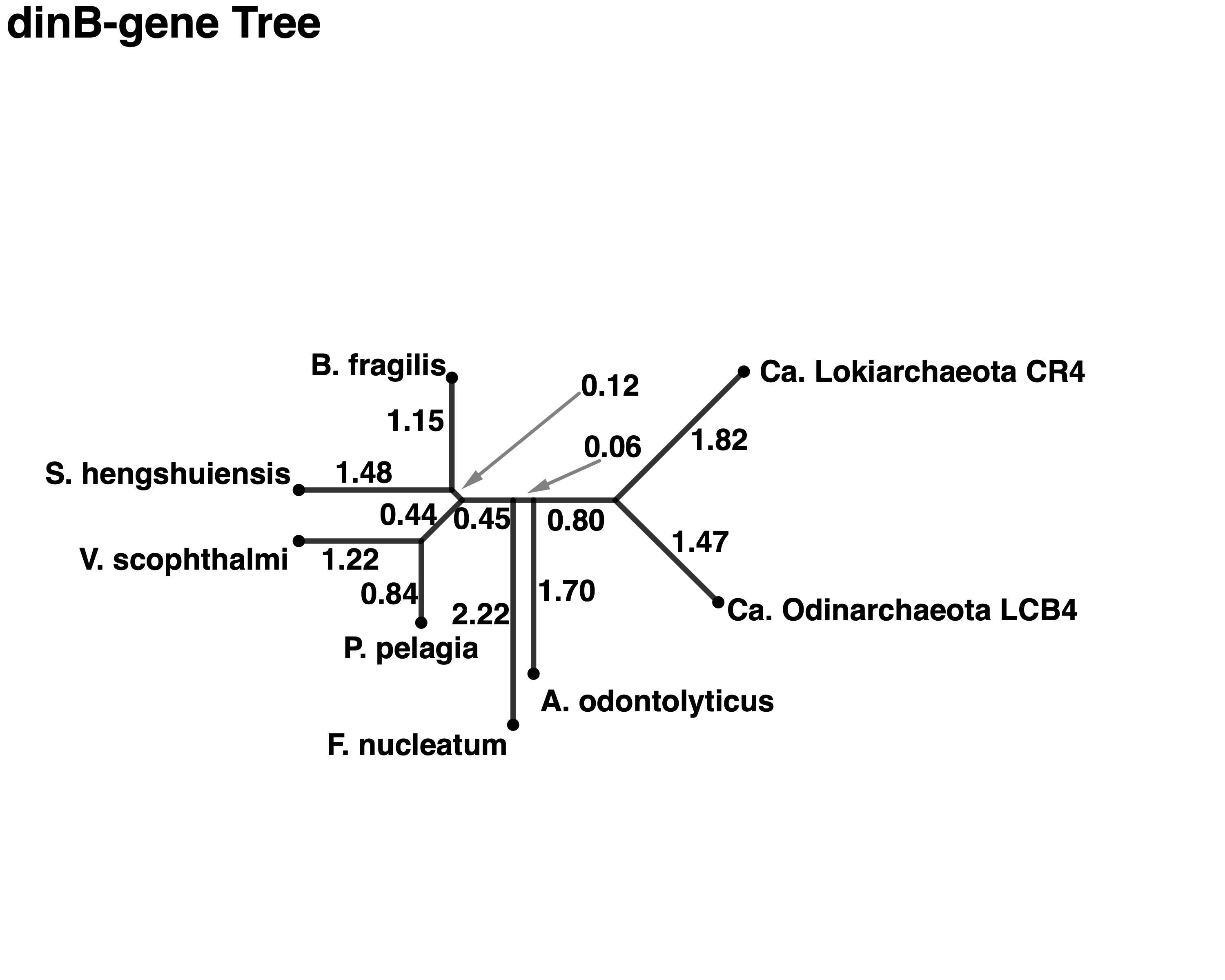}}
    \subfigure[]{\includegraphics[width=0.85\textwidth]{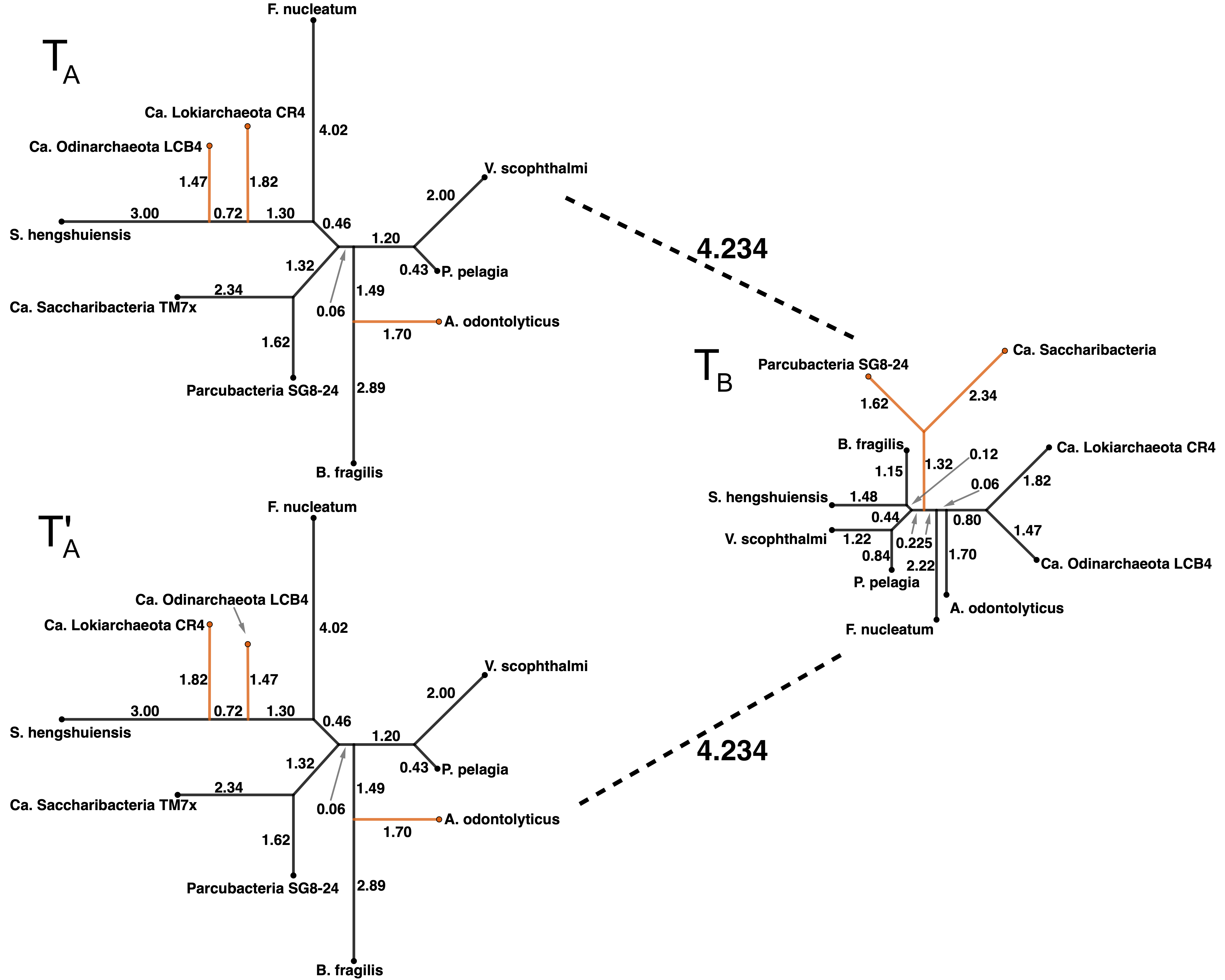}}
    \subfigure[]{\includegraphics[width=0.5\textwidth]{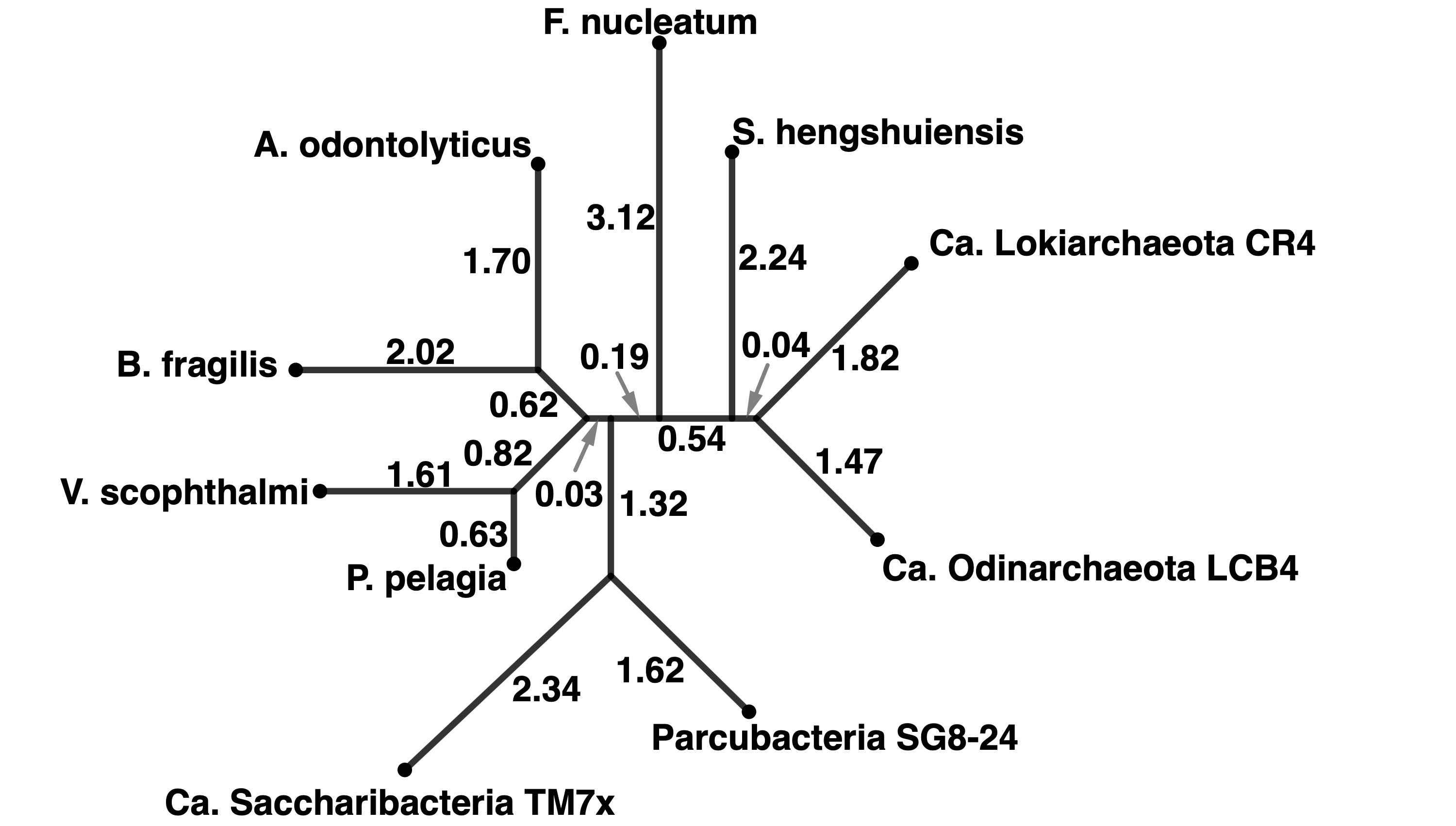}}
    \caption{(a) The estimated evolutionary history of the \textit{ftsA} (left) and \textit{dinB} (right) genes for 10 organisms. (b) The minimal distance between the extension spaces of these trees is 4.234, which can be obtained via two geodesic paths. The two tree pairs $(T_{A},T_{B})$ and $(T'_{A},T_{B})$ that achieve the minimal distance are shown. (c) The midpoint of the geodesics between $(T_{A},T_{B})$ and $(T'_{A},T_{B})$. This midpoint is the same for both geodesics.}
\end{figure}

\section{Discussion}
\label{sec:conclusions}

Extension spaces, first introduced by \cite{GrindstaffGillian2019RoPL}, provide an intuitive approach to contextualizing phylogenetic trees with reduced leaf sets in higher-dimensional BHV spaces. In this paper, we propose to define distances between trees with non-identical leaf sets by finding the minimum distance between their extension spaces, and we developed a reduced gradient algorithm to calculate this distance. A major advantage of this approach is that it gives a measure of dissimilarity that can be applied to \textit{any} two trees. It therefore addresses some of the limitations of the \cite{GrindstaffGillian2019RoPL}, such as that the trees under comparison must share common orthants in their extension spaces, and must have all internal branches of strictly positive length. 

An additional advantage of our approach is that it enables construction of a ``supertree'' that summarizes a pair of trees with respect to topological and edge length differences, even when those trees have non-identical leaf sets. This suggests a measure of compatibility among a collection of trees as the Fr\'echet means of their extension spaces $E_{T_1}^\mathcal{N}, E_{T_2}^\mathcal{N}, \ldots, E_{T_r}^\mathcal{N}$, which for $r=2$ reduces to a midpoint along a geodesic. Note that as minimal distance paths between extension spaces are not necessarily unique, Fr\'echet means of extension spaces are also not necessarily unique. Interestingly, in our applied data example, while there were two minimal distance paths between our trees, both paths had the same midpoint. We conjecture that in applied data examples, Fr\'echet means of extension spaces may often be unique. We leave the construction and study of Fr\'echet means to future work. 

While finding the minimal BHV distance between extension spaces is highly intuitive, our proposed distance is not formally a metric between trees  \citep[Section 3.4]{GrindstaffGillian2019RoPL}.  
Specifically, distinct trees can have intersecting extension spaces  \citep[Example 4.1]{GrindstaffGillian2019RoPL}, and therefore a zero distance, violating positivity. Furthermore, it is possible to find trees $T_1, T_2$ and $T_3$ for which $d(E^{\mathcal{N}}_{T_1}, E^{\mathcal{N}}_{T_3}) = 0$ and $d(E^{\mathcal{N}}_{T_2}, E^{\mathcal{N}}_{T_3}) = 0$, but for which $d(E^{\mathcal{N}}_{T_1}, E^{\mathcal{N}}_{T_2}) > 0$, thus violating the triangle inequality  \citep[Remark 3.6]{GrindstaffGillian2019RoPL}. 
Despite this, our distance still provides a useful measure of similarity between phylogenies with non-identical leaf sets, and when $\mathcal{L}_1 = \mathcal{L}_2 = \mathcal{N}$, our distance reduces to the classical BHV distance. In addition, the algorithm we developed is broadly applicable to the minimization of any convex function over a subset of a BHV space defined by linear constraints (see also \cite{MillerOwenProvan}), which could be broadly useful in other mathematical or computational phylogenetics problems. 

In practice, our algorithm runs within an hour on a modern laptop for up to 10 total leaves without multithreading. Alas, computation time grows quickly in the total number of leaves. For example, if the largest example in Table 1 had one more leaf not included in the second tree (increasing $|\mathcal{L}_1|$ to $9$ and $|\mathcal{L}_1 \cup \mathcal{L}_2|$ to $11$), the number of orthant pairs would grow by a factor of $\sim 22$, from 418,275 to 9,298,575, for which we estimate a single-thread runtime of $\sim 24$ hours. That said, the method can be trivially parallelized across orthant pairs, making it well-suited to distributed computing environments. While we report single-thread runtimes for transparency, our open-source software package implements multi-threading, conveniently accelerating the method for typical (non-distributed) computing environments. 

Because extension spaces can be characterized within a given orthant as a linear system of equations, our algorithm employs a reduced gradient method. Reduced gradient methods are iterative procedures, and therefore their computational complexity is challenging to characterize. That said, in practice, we find that the number of iterations per orthant pair is generally low, with 50\% converging with a gradient of $<10^{-8}$ within 4-6 iterations and 90\% converging within 6-12 iterations. Each iteration, however, involves the computation of multiple 
BHV geodesics. As computing a geodesic is $O(n^4)$-time  \citep[Theorem 3.5]{OwenMegan2011}, each 
iteration of Algorithm \ref{alg:OrthantExtensionDistance} 
is $O(n^4 d)$, where $d$ is the number of iterations required for the line search. 
The number of orthant pairs grows at $O(|\mathcal{N}|^{c})$, 
where $c$ is the number of leaves to be added to the trees,
further contributing to the runtime of the algorithm. Unsurprisingly, in practice, we find that the number of orthant pairs, rather than the geodesic computations, is the major limiting factor in calculating our distance. 
As a result, future work to accelerate computation could consider excluding suboptimal orthant pairs from consideration, such as by excluding highly dissimilar topologies while prioritizing the attachment of inconsequential edges to long edges.
That said, as previously mentioned, orthant pair comparisons can be parallelized across distributed computing architecture, reducing the in-practice computation time by a factor equal to the number of machines available. 



\appendix

\appendixone

\section*{Algorithm: additional details}
\label{sec:AlgDets}

Here we provide additional details pertaining to Algorithm \ref{alg:OrthantExtensionDistance}. 

\textbf{Reduced gradient directions}: 
The direction of change will be based on the gradient of $\delta$ at the current point $\dot{\mathbf{x}}^{t}$ when all partial derivatives in \eqref{equation:PartialDerivativesExpression} are well-defined, and a subgradient otherwise. To compute this (sub)gradient, we first find the support for the geodesic from $T'_1(\dot{\mathbf{x}}^{t})$ to $T'_2(\dot{\mathbf{x}}^{t})$, and use its support to determine the entries of the (sub)gradient. 
We then focus on the reduced gradient within the current facet, 
which is
$\nabla \varphi(\mathbf{x}_{\mathbf{F}}) = \nabla_{\mathbf{F}} f(\mathbf{x}) - A_{\mathbf{F}}^{\top} \left[A_{\mathbf{D}}^{-1}\right]^{\top} \nabla_{\mathbf{D}} f(\mathbf{x})$ for the general case. 
Due to the structure of $\dot{\mathbf{M}}$, 
the partial derivative corresponding to a free variable $\dot{\mathbf{x}}_{j}^{t}$, $j \in \mathbf{F}$, is $\frac{\partial \varphi (\dot{\mathbf{x}}_{\mathbf{F}})}{ \varphi \dot{\mathbf{x}}_{j}^{t}} = \nabla_{j} \delta(\dot{\mathbf{x}}) - \nabla_{j'} \delta(\dot{\mathbf{x}}) $, where $j' \in \mathbf{D}$ is the only index in the dependent variable set such that $\dot{\mathbf{M}}[i,j] = \dot{\mathbf{M}}[i,j'] = 1$. 

There are multiple ways to determine a good direction of change for the free variables at $\dot{\mathbf{x}}_{\mathbf{F}}^{t}$. We employed the conjugate gradient method  \citep[Section 5.5]{NLoptimization} because of its simplicity and the potential gain in efficiency. In this method, the main driver for the direction of change is the (sub)gradient at the current point, but after the first iteration a correction is added to increase efficiency. The correction loses its advantages and new directions become inefficient after many iterations  \citep[Section 5.5.2]{NLoptimization}, thus we re-initialize the correction regularly. The direction of change $\mathbf{d}_{\mathbf{F}}^{t}$ for the free variables is computed by
\begin{equation}
    \label{eq:ConjugateDirection}
    \mathbf{d}_{\mathbf{F}}^{t} = \begin{cases}
        - \nabla \varphi(\dot{\mathbf{x}}^{t}_{\mathbf{F}}) & \text{ if } c_{\text{conj}} = 1\\
        - \nabla \varphi(\dot{\mathbf{x}}^{t}_{\mathbf{F}}) + \frac{\langle \varphi(\dot{\mathbf{x}}^{t}_{\mathbf{F}}), \varphi(\dot{\mathbf{x}}^{t}_{\mathbf{F}}) - \varphi(\dot{\mathbf{x}}^{t-1}_{\mathbf{F}}) \rangle}{\left|\left|\varphi(\dot{\mathbf{x}}^{t-1}_{\mathbf{F}}) \right|\right|^{2}} \mathbf{d}_{\mathbf{F}}^{t-1} & \text{ if } c_{\text{conj}} > 1,
    \end{cases}
\end{equation}
for a counter $c_{\text{conj}}$. In our method, $c_{\text{conj}}$ is re-initialized (reset to 1) every time a new facet is reached (i.e. $\mathbf{F}$ is re-defined and $\mathbf{d}_{\mathbf{F}}^{t-1}$ is no longer of the same dimension) or when $c_{\text{conj}} > 15$. The threshold of 15 was recommended by  \citep[Page 248]{NLoptimization}, and we found it to work well in our setting in practice. 
Given the direction of change for the free variables $\mathbf{d}_{\mathbf{F}}^{t}$, we can find the direction of change for the dependent variables as $\mathbf{d}^{t}_{\mathbf{D}} = - \dot{\mathbf{M}}_{\mathbf{D}}^{-1} \dot{\mathbf{M}}_{\mathbf{F}} \mathbf{d}_{\mathbf{F}}^{t}$ and for the null variables as $\mathbf{d}^{t}_{\mathbf{N}} = 0$. 

\textbf{Selecting step sizes}: 
Given a non-zero direction of change $\mathbf{d}^{t}$, we now discuss finding the best next point on the line segment $\dot{\mathbf{x}}^{t} + \tau \mathbf{d}^{t}$ such that both $\tau \geq 0$ and $\dot{\mathbf{x}}^{t} + \tau \mathbf{d}^{t} \geq 0$. Let $\tau_{\text{max}}$ be the maximum value of $\tau$ such that  $\dot{\mathbf{x}}^{t} + \tau \mathbf{d}^{t} \geq 0$, which is finite because $\dot{\mathbf{M}} \left(\dot{\mathbf{x}}^{t} + \tau \mathbf{d}^{t}\right) = \dot{\mathbf{v}}$ implies $\dot{\mathbf{M}}\mathbf{d}^{t} = 0$, and thus some entry of $\mathbf{d}^{t}$ is negative. 
Thus, we are looking for $\tau \in [0, \tau_{\text{max}}]$ that minimizes $\delta(\dot{\mathbf{x}}^{t} + \tau \mathbf{d}^{t})$. 


Using the chain-rule, the derivative of $\delta(\dot{\mathbf{x}}^{t} + \tau \mathbf{d}^{t})$ with respect to $\tau$ is 
\begin{equation}
    h(\tau) = \frac{\partial \delta(\dot{\mathbf{x}}^{t} + \tau \mathbf{d}^{t})}{ \partial \tau} = \langle \nabla \delta(\dot{\mathbf{x}}^{t} + \tau \mathbf{d}^{t}), \mathbf{d}^{t} \rangle.
\end{equation}
The function $\tau \mapsto \delta(\dot{\mathbf{x}}^{t} + \tau \mathbf{d}^{t})$ is convex and by construction of $\mathbf{d}^t$, $h(0) < 0$. 
From this, we employ a derivative-based bisection method to reach the minimum of this function. We start by checking the value of $h(\tau_{\text{max}})$. If $h(\tau_{\text{max}}) \leq 0$, then we have reached the minimum at $\tau_{0} = \tau_{\text{max}}$. Otherwise, we search for the value $\tau_{0} \in [0, \tau_{\max}]$ such that $h(\tau_{0}) = 0$ as follows 

\begin{enumerate}
    \item Initialize $\tau_{\text{left}} = 0$ and $\tau_{\text{right}} = \tau_{\text{max}}$.
    \item Take $\tau^{*} = \frac{\tau_{\text{left}} +  \tau_{\text{max}}}{2}$. 
    \item Evaluate $h(\tau^{*})$: 
    \begin{itemize}
        \item If $h(\tau^{*}) = 0$, return $\tau^{*}$.
        \item If $h(\tau^{*}) < 0$, update $\tau_{\text{left}} = \tau^{*}$ and return to step 2. 
        \item If $h(\tau^{*}) > 0$, update $\tau_{\text{right}} = \tau^{*}$ and return to step 2. 
    \end{itemize}
\end{enumerate}
In practice, we do not require $h(\tau^{*}) = 0$ exactly; instead, we require $|h(\tau^{*})|$ to be below a threshold. We find that $|h(\tau^{*})| < 10^{-16}$ works well in practice. 
This threshold can be altered in our software via the flag \texttt{-Tol2}.
After finding $\tau_{0}$, we select the next point as $\dot{\mathbf{x}}^{t+1} = \dot{\mathbf{x}}^{t} + \tau_{0} \mathbf{d}^{t}$.


\textbf{Thresholds for convergence}: 
Our iterative algorithm is guaranteed to converge to stationary points with zero (sub)gradients. However, in practice, we employ a tolerance threshold to find a solution that is sufficiently close to a stationary point in a finite number of iterations. 
Specifically, before computing a new direction of change, we test 
if we have reached the global minimum by checking if every entry in $\nabla \varphi(\dot{\mathbf{x}}_{\mathbf{F}})$ is less than a given threshold \texttt{Tol1}. In practice, we find that choosing this threshold to be $10^{-8}$  produces good performance. The user can alter this value using the flag \texttt{-Tol1}.
\bibliographystyle{biometrika}

\bibliography{references}

\begin{thebibliography}{32}
\expandafter\ifx\csname natexlab\endcsname\relax\def\natexlab#1{#1}\fi

\bibitem[{Barden et~al.(2013)Barden, Le \& Owen}]{blow}
\textsc{Barden, D.}, \textsc{Le, H.} \& \textsc{Owen, M.} (2013).
\newblock {Central limit theorems for Fr\'echet means in the space of
  phylogenetic trees}.
\newblock \textit{Electron. J. Probab} \textbf{18}, 1--25.

\bibitem[{Benner et~al.(2014)Benner, Bac{\'a}k \& Bourguignon}]{bbb}
\textsc{Benner, P.}, \textsc{Bac{\'a}k, M.} \& \textsc{Bourguignon, P.-Y.}
  (2014).
\newblock {Point estimates in phylogenetic reconstructions}.
\newblock \textit{Bioinformatics} \textbf{30}, 534--540.

\bibitem[{Billera et~al.(2001)Billera, Holmes \& Vogtmann}]{BILLERA2001}
\textsc{Billera, L.~J.}, \textsc{Holmes, S.~P.} \& \textsc{Vogtmann, K.}
  (2001).
\newblock Geometry of the space of phylogenetic trees.
\newblock \textit{Advances in Applied Mathematics} \textbf{27}, 733--767.

\bibitem[{Boyd \& Vandenberghe(2004)}]{alma99124373890001452}
\textsc{Boyd, S.~P.} \& \textsc{Vandenberghe, L.} (2004).
\newblock \textit{Convex optimization}.
\newblock Cambridge, UK: Cambridge University Press.

\bibitem[{Bridson \& Haefliger(1999)}]{Bridson:1999ky}
\textsc{Bridson, M.~R.} \& \textsc{Haefliger, A.} (1999).
\newblock \textit{{Metric Spaces of Non-Positive Curvature}}, vol. 319 of
  \textit{Grundlehren der mathematischen Wissenschaften}.
\newblock Berlin, Heidelberg: Springer Berlin Heidelberg.

\bibitem[{Brown \& Owen(2020)}]{brown2020mean}
\textsc{Brown, D.~G.} \& \textsc{Owen, M.} (2020).
\newblock Mean and variance of phylogenetic trees.
\newblock \textit{Systematic biology} \textbf{69}, 139--154.

\bibitem[{Chen et~al.(2016)Chen, Burall, Luo, Timme, Melka, Muruvanda, Payne,
  Wang, Kastanis, Maounounen-Laasri et~al.}]{chen2016listeria}
\textsc{Chen, Y.}, \textsc{Burall, L.~S.}, \textsc{Luo, Y.}, \textsc{Timme,
  R.}, \textsc{Melka, D.}, \textsc{Muruvanda, T.}, \textsc{Payne, J.},
  \textsc{Wang, C.}, \textsc{Kastanis, G.}, \textsc{Maounounen-Laasri, A.}
  et~al. (2016).
\newblock Listeria monocytogenes in stone fruits linked to a multistate
  outbreak: enumeration of cells and whole-genome sequencing.
\newblock \textit{Applied and Environmental Microbiology} \textbf{82},
  7030--7040.

\bibitem[{Dagan \& Martin(2006)}]{dagan2006tree}
\textsc{Dagan, T.} \& \textsc{Martin, W.} (2006).
\newblock The tree of one percent.
\newblock \textit{Genome biology} \textbf{7}, 1--7.

\bibitem[{Duarte et~al.(2020)Duarte, Ngugi, Alam, Pearman, Kamau, Eguiluz,
  Gojobori, Acinas, Gasol, Bajic et~al.}]{duarte2020sequencing}
\textsc{Duarte, C.~M.}, \textsc{Ngugi, D.~K.}, \textsc{Alam, I.},
  \textsc{Pearman, J.}, \textsc{Kamau, A.}, \textsc{Eguiluz, V.~M.},
  \textsc{Gojobori, T.}, \textsc{Acinas, S.~G.}, \textsc{Gasol, J.~M.},
  \textsc{Bajic, V.} et~al. (2020).
\newblock Sequencing effort dictates gene discovery in marine microbial
  metagenomes.
\newblock \textit{Environmental Microbiology} \textbf{22}, 4589--4603.

\bibitem[{Gori et~al.(2016)Gori, Suchan, Alvarez, Goldman \&
  Dessimoz}]{gori2016clustering}
\textsc{Gori, K.}, \textsc{Suchan, T.}, \textsc{Alvarez, N.}, \textsc{Goldman,
  N.} \& \textsc{Dessimoz, C.} (2016).
\newblock Clustering genes of common evolutionary history.
\newblock \textit{Molecular biology and evolution} \textbf{33}, 1590--1605.

\bibitem[{Grindstaff \& Owen(2019)}]{GrindstaffGillian2019RoPL}
\textsc{Grindstaff, G.} \& \textsc{Owen, M.} (2019).
\newblock Representations of partial leaf sets in phylogenetic tree space.
\newblock \textit{SIAM journal on applied algebra and geometry} \textbf{3},
  691--720.

\bibitem[{Imachi et~al.(2020)Imachi, Nobu, Nakahara, Morono, Ogawara, Takaki,
  Takano, Uematsu, Ikuta, Ito et~al.}]{imachi2020isolation}
\textsc{Imachi, H.}, \textsc{Nobu, M.~K.}, \textsc{Nakahara, N.},
  \textsc{Morono, Y.}, \textsc{Ogawara, M.}, \textsc{Takaki, Y.},
  \textsc{Takano, Y.}, \textsc{Uematsu, K.}, \textsc{Ikuta, T.}, \textsc{Ito,
  M.} et~al. (2020).
\newblock Isolation of an archaeon at the prokaryote--eukaryote interface.
\newblock \textit{Nature} \textbf{577}, 519--525.

\bibitem[{Maddison(1997)}]{Maddison:1997de}
\textsc{Maddison, W.~P.} (1997).
\newblock {Gene Trees in Species Trees}.
\newblock \textit{Systematic Biology} \textbf{46}, 523--536.

\bibitem[{Miller et~al.(2015)Miller, Owen \& Provan}]{MillerOwenProvan}
\textsc{Miller, E.}, \textsc{Owen, M.} \& \textsc{Provan, J.~S.} (2015).
\newblock Polyhedral computational geometry for averaging metric phylogenetic
  trees.
\newblock \textit{Advances in applied mathematics} \textbf{68}, 51--91.

\bibitem[{Nye(2011)}]{nye11}
\textsc{Nye, T.~M.} (2011).
\newblock {Principal components analysis in the space of phylogenetic trees}.
\newblock \textit{The Annals of Statistics} \textbf{39}, 2716--2739.

\bibitem[{Nye et~al.(2017)Nye, Tang, Weyenberg \& Yoshida}]{nyepca}
\textsc{Nye, T. M.~W.}, \textsc{Tang, X.}, \textsc{Weyenberg, G.} \&
  \textsc{Yoshida, R.} (2017).
\newblock {Principal component analysis and the locus of the Fréchet mean in
  the space of phylogenetic trees}.
\newblock \textit{Biometrika} \textbf{104}, 901--922.

\bibitem[{Owen(2011)}]{owen1}
\textsc{Owen, M.} (2011).
\newblock {Computing geodesic distances in tree space}.
\newblock \textit{SIAM Journal on Discrete Mathematics} \textbf{25},
  1506--1529.

\bibitem[{Owen \& Provan(2011)}]{OwenMegan2011}
\textsc{Owen, M.} \& \textsc{Provan, J.~S.} (2011).
\newblock A fast algorithm for computing geodesic distances in tree space.
\newblock \textit{IEEE/ACM transactions on computational biology and
  bioinformatics} \textbf{8}, 2--13.

\bibitem[{Parks et~al.(2020)Parks, Chuvochina, Chaumeil, Rinke, Mussig \&
  Hugenholtz}]{ParksDonovanH2020Acdt}
\textsc{Parks, D.~H.}, \textsc{Chuvochina, M.}, \textsc{Chaumeil, P.-A.},
  \textsc{Rinke, C.}, \textsc{Mussig, A.~J.} \& \textsc{Hugenholtz, P.} (2020).
\newblock A complete domain-to-species taxonomy for bacteria and archaea.
\newblock \textit{Nature biotechnology} \textbf{38}, 1079--.

\bibitem[{Ren et~al.(2017)Ren, Zha, Bi, Sanchez, Monical, Delcourt, Guzman \&
  Davidson}]{ren2017combinatorial}
\textsc{Ren, Y.}, \textsc{Zha, S.}, \textsc{Bi, J.}, \textsc{Sanchez, J.~A.},
  \textsc{Monical, C.}, \textsc{Delcourt, M.}, \textsc{Guzman, R.~K.} \&
  \textsc{Davidson, R.} (2017).
\newblock A combinatorial method for connecting bhv spaces representing
  different numbers of taxa.
\newblock \textit{arXiv preprint arXiv:1708.02626v2} .

\bibitem[{Rinke et~al.(2013)Rinke, Schwientek, Sczyrba, Ivanova, Anderson,
  Cheng, Darling, Malfatti, Swan, Gies, Dodsworth, Hedlund, Tsiamis, Sievert,
  Liu, Eisen, Hallam, Kyrpides, Stepanauskas, Rubin, Hugenholtz \&
  Woyke}]{RinkeChristian2013Iitp}
\textsc{Rinke, C.}, \textsc{Schwientek, P.}, \textsc{Sczyrba, A.},
  \textsc{Ivanova, N.~N.}, \textsc{Anderson, I.~J.}, \textsc{Cheng, J.-F.},
  \textsc{Darling, A.}, \textsc{Malfatti, S.}, \textsc{Swan, B.~K.},
  \textsc{Gies, E.~A.}, \textsc{Dodsworth, J.~A.}, \textsc{Hedlund, B.~P.},
  \textsc{Tsiamis, G.}, \textsc{Sievert, S.~M.}, \textsc{Liu, W.-T.},
  \textsc{Eisen, J.~A.}, \textsc{Hallam, S.~J.}, \textsc{Kyrpides, N.~C.},
  \textsc{Stepanauskas, R.}, \textsc{Rubin, E.~M.}, \textsc{Hugenholtz, P.} \&
  \textsc{Woyke, T.} (2013).
\newblock Insights into the phylogeny and coding potential of microbial dark
  matter.
\newblock \textit{Nature (London)} \textbf{499}, 431--437.

\bibitem[{Royalty \& Steen(2019)}]{royalty2019theoretical}
\textsc{Royalty, T.~M.} \& \textsc{Steen, A.~D.} (2019).
\newblock Theoretical and simulation-based investigation of the relationship
  between sequencing effort, microbial community richness, and diversity in
  binning metagenome-assembled genomes.
\newblock \textit{Msystems} \textbf{4}, 10--1128.

\bibitem[{Ruszczy\'nski(2006)}]{NLoptimization}
\textsc{Ruszczy\'nski, A.~P.} (2006).
\newblock \textit{Nonlinear optimization}.
\newblock Princeton, N.J: Princeton University Press.

\bibitem[{Scaduto et~al.(2010)Scaduto, Brown, Haaland, Zwickl, Hillis \&
  Metzker}]{scaduto2010source}
\textsc{Scaduto, D.~I.}, \textsc{Brown, J.~M.}, \textsc{Haaland, W.~C.},
  \textsc{Zwickl, D.~J.}, \textsc{Hillis, D.~M.} \& \textsc{Metzker, M.~L.}
  (2010).
\newblock Source identification in two criminal cases using phylogenetic
  analysis of hiv-1 dna sequences.
\newblock \textit{Proceedings of the National Academy of Sciences}
  \textbf{107}, 21242--21247.

\bibitem[{Sturm(2003)}]{sturm}
\textsc{Sturm, K.-T.} (2003).
\newblock {Probability measures on metric spaces of nonpositive curvature}.
\newblock \textit{Contemporary mathematics} \textbf{338}, 357--390.

\bibitem[{Teichman et~al.(2023)Teichman, Lee \& Willis}]{teichman2023analyzing}
\textsc{Teichman, S.}, \textsc{Lee, M.~D.} \& \textsc{Willis, A.~D.} (2023).
\newblock Analyzing microbial evolution through gene and genome phylogenies.
\newblock \textit{bioRxiv} .

\bibitem[{Weyenberg et~al.(2014)Weyenberg, Huggins, Schardl, Howe \&
  Yoshida}]{weyenberg2014kdetrees}
\textsc{Weyenberg, G.}, \textsc{Huggins, P.~M.}, \textsc{Schardl, C.~L.},
  \textsc{Howe, D.~K.} \& \textsc{Yoshida, R.} (2014).
\newblock Kdetrees: non-parametric estimation of phylogenetic tree
  distributions.
\newblock \textit{Bioinformatics} \textbf{30}, 2280--2287.

\bibitem[{Willis(2019)}]{willis2019confidence}
\textsc{Willis, A.} (2019).
\newblock Confidence sets for phylogenetic trees.
\newblock \textit{Journal of the American Statistical Association}
  \textbf{114}, 235--244.

\bibitem[{Willis \& Bell(2018)}]{willis2018uncertainty}
\textsc{Willis, A.} \& \textsc{Bell, R.} (2018).
\newblock Uncertainty in phylogenetic tree estimates.
\newblock \textit{Journal of Computational and Graphical Statistics}
  \textbf{27}, 542--552.

\bibitem[{Zaheer et~al.(2018)Zaheer, Noyes, Ortega~Polo, Cook, Marinier,
  Van~Domselaar, Belk, Morley \& McAllister}]{zaheer2018impact}
\textsc{Zaheer, R.}, \textsc{Noyes, N.}, \textsc{Ortega~Polo, R.},
  \textsc{Cook, S.~R.}, \textsc{Marinier, E.}, \textsc{Van~Domselaar, G.},
  \textsc{Belk, K.~E.}, \textsc{Morley, P.~S.} \& \textsc{McAllister, T.~A.}
  (2018).
\newblock Impact of sequencing depth on the characterization of the microbiome
  and resistome.
\newblock \textit{Scientific reports} \textbf{8}, 5890.

\bibitem[{Zairis et~al.(2016)Zairis, Khiabanian, Blumberg \&
  Rabadan}]{zairis2016genomic}
\textsc{Zairis, S.}, \textsc{Khiabanian, H.}, \textsc{Blumberg, A.~J.} \&
  \textsc{Rabadan, R.} (2016).
\newblock {Genomic data analysis in tree spaces}.
\newblock \textit{arXiv:1607.07503} .

\bibitem[{Zhu et~al.(2019)Zhu, Mai, Pfeiffer, Janssen, Asnicar, Sanders,
  Belda-Ferre, Al-Ghalith, Kopylova, McDonald et~al.}]{zhu2019phylogenomics}
\textsc{Zhu, Q.}, \textsc{Mai, U.}, \textsc{Pfeiffer, W.}, \textsc{Janssen,
  S.}, \textsc{Asnicar, F.}, \textsc{Sanders, J.~G.}, \textsc{Belda-Ferre, P.},
  \textsc{Al-Ghalith, G.~A.}, \textsc{Kopylova, E.}, \textsc{McDonald, D.}
  et~al. (2019).
\newblock Phylogenomics of 10,575 genomes reveals evolutionary proximity
  between domains {B}acteria and {A}rchaea.
\newblock \textit{Nature Communications} \textbf{10}, 1--14.

\end{thebibliography}
\end{document}